\documentclass[aps, pra, a4paper, twocolumn, 10pt, showpacs,superscriptaddress]{revtex4}
\usepackage{bbm, amsmath, amssymb, amsthm, bm,textcomp, nicefrac, mathtools}
\usepackage{amsthm,amssymb,amsmath,graphicx,epsfig,color,verbatim,enumerate,amsthm}
\input{Qcircuit}

\newcommand{\bracket}[2]{\left\langle #1|#2\right\rangle}
\newcommand\defn[1]{\textsl{#1}}

\newcommand\ketbra[1]{|#1\rangle\langle#1|}
\newcommand\cH{{\mathcal H}}
\newcommand\cM{{\mathcal M}}
\newcommand\cN{{\mathcal N}}
\newcommand\cR{{\mathcal R}}

\newcommand\cD{{\mathcal D}}
\newcommand\cI{{\mathcal I}}
\newcommand\cB{{\mathcal B}}
\newcommand\cE{{\mathcal E}}

\newcommand\cP{{\mathcal P}}
\def\sx{\sigma_x}
\def\sy{\sigma_y}
\def\sz{\sigma_z}

\def\sx{\sigma_x}
\def\sy{\sigma_y}
\def\sz{\sigma_z}

\newcommand{\bea}{\begin{eqnarray}}
\newcommand{\eea}{\end{eqnarray}}
\def\bi{\begin{itemize}}
\def\ei{\end{itemize}}
\def\bc{\begin{center}}
\def\ec{\end{center}}

\def\C{\hbox{$\mit I$\kern-.7em$\mit C$}}
\def\R{\hbox{$\mit I$\kern-.6em$\mit R$}}

\def\ket#1{|#1\rangle}
\newcommand{\one}{\mbox{$1 \hspace{-1.0mm}  {\bf l}$}}

\def\ket#1{\left| #1\right>}
\def\bra#1{\left< #1\right|}

\newcommand{\proj}[1]{\ket{#1}\bra{#1}}

\makeatletter
\newtheorem*{rep@theorem}{\rep@title}
\newcommand{\newreptheorem}[2]{%
\newenvironment{rep#1}[1]{%
 \def\rep@title{#2 \ref{##1}}%
 \begin{rep@theorem}}%
 {\end{rep@theorem}}}
\makeatother

\newreptheorem{theorem}{Theorem}
\newtheorem{lemma}{Lemma}

\newtheorem{theorem}{Theorem}

\begin{document}
\title{Efficient quantum communication under collective noise}
\author{Michael Skotiniotis}
\affiliation{Institute
for Quantum Information Science, University of Calgary,2500 University
Drive NW,Calgary AB, T2l 1N1, Canada}
\affiliation{Institut f\"ur Theoretische Physik, Universit\"at
  Innsbruck, Technikerstr. 25, A-6020 Innsbruck,  Austria.}
\author{Wolfgang D\"{u}r}
\affiliation{Institut f\"ur Theoretische Physik, Universit\"at
  Innsbruck, Technikerstr. 25, A-6020 Innsbruck,  Austria.}
\author{Barbara Kraus}
\affiliation{Institut f\"ur Theoretische Physik, Universit\"at
  Innsbruck, Technikerstr. 25, A-6020 Innsbruck,  Austria.}
\date{\today}
\begin{abstract}
We introduce a new quantum communication protocol for the transmission of quantum information under collective noise. Our protocol utilizes a decoherence-free subspace in such a way that an optimal asymptotic transmission rate is achieved, while at the same time encoding and decoding operations can be efficiently implemented. The encoding and decoding circuit requires a number of elementary gates that scale linearly with the number of transmitted qudits, $m$.  The logical depth of our encoding and decoding operations is constant and depends only on the channel in question. For channels described by an arbitrary discrete group $G$, i.e.~with a discrete number, $\lvert G\rvert$, of possible noise operators, perfect transmission at a rate $m/(m+r)$ is achieved with an overhead that scales at most as $\mathcal{O}(d^r)$ where the number of auxiliary qudits, $r$, depends solely on the group in question. Moreover, this overhead is independent of the number of transmitted qudits, $m$.  For certain groups, e.g.~cyclic groups, we find that the overhead scales only linearly with the number of group elements $|G|$.  
\end{abstract}
\maketitle
\section{\label{sec:1}Introduction}

The transmission of quantum information between several communication partners is a central element of many applications of quantum information theory, including quantum cryptography and quantum key distribution~\cite{BB84}, quantum networks~\cite{CZKM97,K08}, and distributed quantum computation~\cite{CEHM99}. In a realistic set-up quantum communication is subject to noise and imperfections leading to imperfect communication channels. Several methods are known to deal with such a situation, including teleportation-based communication utilizing entanglement purification \cite{Dur07}, or encoding of quantum information using quantum error-correcting codes~\cite{Sh95,St96a,St96b,BDSW96,KL97,EM96,LMPZ96,CS96,Go96,CRSS97,CRSS98}. In cases where only restricted types of errors occur, i.e.~when parties lack a shared phase or spatial reference frame (see~\cite{BRS07} and references therein), or when quantum communication is subject to collective noise, error-avoiding schemes are also available~\cite{ZR98, Zan00,LCW98,BKLW00,LBW99,KLV99,KBLW01,VFPKLC01}. In such schemes information is stored and transmitted in a decoherence-free subspace (DFS).
A similar problem occurs when considering the storage of quantum information where the source of collective noise is, for example, a globally fluctuating magnetic field.

In this paper we propose a novel quantum communication and storage scheme capable of dealing with collective noise associated to an arbitrary finite group $G$. The scheme is efficient with respect to the required encoding and decoding operations, and is capable of achieving an optimal transmission or storage rate in the asymptotic limit. Our scheme is mainly based on the usage of a DFS, however elements from standard error-correcting codes are also utilized. Whereas error-avoiding schemes encode quantum information directly into a DFS, error-correcting schemes usually involve measurements of ancilla systems whose outcome allows one to determine the kind of error that occurred and deduce the required correction operations.  We use ancilla systems for the storage or transmission of $m$ qudits in such a way that we construct a joint state of system plus ancilla that lies within a DFS. Similar to standard error correction schemes we measure the ancilla system and use the measurement outcome to determine the required correction operation. However, in our scheme the measurement does not reveal any information about the channel or the kind of error that occurred. In addition, the correction operation is local, i.e.~consists of single-system operations only, making our protocol suitable in multiparty communication scenarios with multiple receivers.

Our protocol is an extension of the measure and re-align protocol proposed in~\cite{BRST09} for communicating quantum information between parties lacking a shared frame of reference.  Such a protocol was shown to be optimal in~\cite{BCDFP10}. Contrary to the protocol in~\cite{BRST09}, where Bob measures an auxiliary system to learn the relative transformation between his and Alice's frame of reference, our protocol hides Alice's direction from Bob while still allowing perfect communication of quantum information.  

We show that for channels described by an arbitrary discrete group $G$, i.e.~with a discrete number, $\lvert G\rvert$, of possible noise operators, the number of ancillary systems, $r$, is independent of the number of qudits, $m$, to be transmitted. Our protocol achieves a transmission rate of $m/(m+r)$, i.e.~$m+r$ physical qudits need to be send through a noisy channel to faithfully transmit $m$ logical qudits, that is asymptotically optimal approaching unity in the large $m$-limit. Furthermore, we provide an explicit encoding and decoding circuit that is efficient. In particular, the required number of single and two-qubit gates (henceforth referred to as elementary gates) for any finite group $G$ scales as $\mathcal{O}(m,|G|\log(|G|), d^r)$. The additional overhead of $\mathcal{O}(d^r)$ gates is required in order to prepare the joint state of system plus ancilla in a DFS, where $r$ is some (finite) integer that depends solely on the channel in question~\footnote{Note that $r$ itself depends on the number of group elements and the specific group in question, which may lead to a worse than linear scaling in $|G|$.}.  In addition, the logical depth of the circuit, the amount of time required by the circuit to generate the desired state, is $\mathcal{O}(|G|\log(|G|),d^r)$, independent of $m$. For the case where $G$ is a cyclic group this overhead can be shown to scale only linearly ($\mathcal{O}(|G|)$) with the number of group elements. 

Since the number of elementary gates required to implement our protocol scales linearly with the number, $m$, of logical qudits our protocol is more efficient than the best error-avoiding communication scheme~\cite{BCH06}.  The latter  scales as $N\mathrm{poly}(\log(N),d)$, where $N$ is the total number of $d$-dimensional physical systems. Recently, an alternative implementation for encoding and decoding in a DFS was proposed that also scales linearly in the number of logical qubits being transmitted, but achieves an asymptotic rate of $1/2$~\cite{LNPST11}~\footnote{Private communication with Mikio Nakahara.}.  Due to the asymptotically optimal rate of transmission and linear scaling of our protocol a practical implementation seems feasible.

This paper is organized as follows. In Sec.~\ref{sec:2} we review the concept of DFS and some basic results within group representation theory. We also introduce our notation in this section and review previous encoding schemes utilizing a DFS~\cite{BCH06}. In Sec.~\ref{sec:3} we present a novel method to encode information in a DFS for collective noise described by an arbitrary finite group, $G$, where we make use of a finite number, $r$, of auxiliary systems. In Sec.~\ref{sec:4} we explicitly construct encoding and decoding circuit implementations for our protocol and discuss the required resources, i.e.~the number of elementary gates, and the logical depth of these circuits. A direct comparison with other methods~\cite{BCH06,LNPST11} shows that our scheme is more efficient. We also consider the transmission rate of our protocol and show its asymptotic optimality. 
We summarize and conclude in Sec. \ref{Conclusion}.

\section{Decoherence-Free Subspace}\label{sec:2}

In this section we review some basic results of group representation theory, in particular the concept of a DFS. Moreover, we outline the basic principles underpinning the best known quantum communication protocols which utilize DFS~\cite{BCH06}.

The problem we are considering is the transmission of quantum information through a quantum channel with collective noise.  One party, the sender, prepares $N$, $d$-dimensional quantum systems in some state, $\rho\in\cB(\cH_d^{\otimes N})$, and sends them through a noisy quantum channel to one or more parties, the receivers. The noise of the quantum channel is described by a set of $M$ operators, $\{U_{g_i},\, i\in(0,\ldots,M-1)\}$, and a probability distribution, $\{p_{g_i}\}$ with $p_{g_i}>0, \forall i\in (0,\ldots,M-1)$ where $\sum_{i=0}^{M-1} p_{g_i}=1$. The noise of the channel is assumed to be {\em collective}.  That is, the same noise operator acts on each of the transmitted $d$-dimensional systems.  After transmission through the channel the quantum state of the $N$ systems is given by
\begin{eqnarray}\nonumber
\cE[\rho]&=&\sum_{i=0}^{M-1} p_{g_i} \left(U^{(1)}_{g_i}\otimes U^{(2)}_{g_i}\otimes\cdots\otimes U^{(N)}_{g_i}\right)[\rho]\\
&&\times\left(U^{(1)}_{g_i}\otimes U^{(2)}_{g_i}\otimes\cdots\otimes U^{(N)}_{g_i}\right)^{\dagger},
\label{1}
\end{eqnarray}
where $U^{(k)}_{g_i}$ denotes the operator $U_{g_i}$ acting on the $k^{\mathrm{th}}$ quantum system. Hence, the receiver(s) obtain corrupt quantum data. Our task is to construct a communication protocol that allows for efficient, error-free communication of quantum information through quantum channels subject to collective noise.

The fact that the noise of the channel is collective allows for the construction of \defn{error-avoiding} protocols~\cite{ZR98, Zan00,LCW98,BKLW00,LBW99,KLV99,KBLW01,VFPKLC01}, i.e.~protocols that can protect quantum information from being corrupted, as we now review.   We can assume, without loss of generality, that the set of operations, $\{U_{g_i},\, i\in(0,\ldots,M-1)\}$, forms a \emph{ unitary representation} of a symmetry group, $G$, acting on the $d$-dimensional Hilbert space, $\cH_d$~\footnote{If the set $\{U_{g_i},\,i\in(0,\ldots, M-1)\}$ does not form a group we can always supplement it with additional operations whose probability of occurrence is zero.}.  We focus only on unitary representations as it is known that every representation of a finite group is equivalent to a unitary representation~\cite{Sternberg}. Two representations, $U$ and $T$, of a group $G$ are \defn{equivalent} if there exists an invertible matrix, $V$, such that for all $g_i\in G$, $VU_{g_i}V^\dagger=T_{g_i}$.  Due to Schur's lemmata~\cite{Sternberg}, the collective representation, $U_{g_i}^{\otimes N}\equiv U^{(1)}_{g_i}\otimes\cdots\otimes U^{(N)}_{g_i},\, g_i\in G$, can be decomposed into irreducible representations (irreps) of $G$, $U^{(\lambda)}$, as follows
\begin{equation}
U_{g_i}^{\otimes N}=\bigoplus_\lambda \alpha^{(\lambda)}\, U_{g_i}^{(\lambda)}, \, \forall g_i\in G.
\label{2}
\end{equation}
Here, $\lambda$ labels the inequivalent irreps of $G$, and $\alpha^{(\lambda)}$ denotes the multiplicity of irrep $U^{(\lambda)}$. Let us denote by $\{\ket{\lambda,m,\beta}\}$ the orthonormal basis in which all matrices $U_{g_i}^{\otimes N},\, i\in(0,\ldots,M-1)$, are block diagonal. As before, $\lambda$ labels the inequivalent irreps of $G$, and $m$ labels an orthonormal basis of the space, $\cM^{(\lambda)}$, on which $U_{g_i}^{(\lambda)}$ acts upon.  Note that the dimension of $\cM^{(\lambda)}$ coincides with the dimension of the irrep $U^{(\lambda)}$. Moreover, $\beta$ labels an orthonormal basis of $\cN^{(\lambda)}\equiv \C^{\alpha^{(\lambda)}}$, the space of dimension $\alpha^{(\lambda)}$, on which $\{U_{g_i}^{\otimes N},\, g_i\in G\}$ acts trivially for any $g_i\in G$.

The total Hilbert space, $\cH_d^{\otimes N}$, can be conveniently written with respect to the block diagonal basis, $\{\ket{\lambda,m,\beta}\}$, as
\begin{equation}
\cH_d^{\otimes N}=\bigoplus_\lambda \cH^{(\lambda)}=\bigoplus_\lambda \cM^{(\lambda)}\otimes \cN^{(\lambda)}.
\label{3}
\end{equation}
Writing an arbitrary state, $\ket{\psi}\in\cH_d^{\otimes N}$, in the basis $\{\ket{\lambda,m,\beta}\}$, i.e.
\begin{equation}
\ket{\psi}=\sum_{\lambda,m,\beta} v_{\lambda,m,\beta} \ket{\lambda,m,\beta},
\label{4}
\end{equation}
where $v_{\lambda,m,\beta}\in\C$ satisfy $\sum_{\lambda,m,\beta}\lvert v_{\lambda,m,\beta}\rvert^2=1$, and acting on this state with the operator $U_{g_i}^{\otimes N}$ yields
\begin{equation}
U_{g_i}^{\otimes N}\ket{\psi}=\sum_{\lambda,m,m',\beta} v_{\lambda,m,\beta}\, u^{(\lambda)}_{m',m}(g_i)\ket{\lambda,m',\beta},
\label{5}
\end{equation}
where $u^{(\lambda)}_{m',m}(g_i)\equiv \bra{\lambda, m'}U_{g_i}^{\otimes N}\ket{\lambda, m}$.  Therefore, the action of the collective noise operations, $\{U_{g_i}^{\otimes N},\, g_i\in G\}$, affects the index associated to the {\em carrier spaces}, $\cM^{(\lambda)}$, but not the index associated to the {\em multiplicity spaces}, $\cN^{(\lambda)}$.
 
Using the results reviewed above it can be shown that, if the probability distribution, $\{p_{g_i}\}$, is the uniform prior Eq.~\eqref{1} reduces to
\begin{equation}
\cE[\rho]=\sum_\lambda \left(\cD_{\cM^{(\lambda)}}\otimes \cI_{\cN^{(\lambda)}}\right)\circ \cP^{(\lambda)}[\rho],
\label{6}
\end{equation}
where $\cD$ is the completely depolarizing map, $\cD(A)=\frac{\mbox{tr}(A)}{\mathrm{dim}(\cH)} \one,\forall A\in\cB(\cH)$, $\cI$ is the identity map, and $\cP^{(\lambda)}(A)=\Pi_\lambda A\Pi_\lambda$, where $\Pi_\lambda$ is the projector onto the space $\cH^{(\lambda)}$~\footnote{If $\{p_{g_i}\}$ is not the uniform prior then the map $\cD$ is a partially decohering map.}.  Such a situation is encountered, for instance, if the collective noise is due to the complete lack of a shared frame of reference associated with the symmetry group $G$~\cite{BRS07}.

The discussion above shows that in the presence of collective noise the total Hilbert space, $\cH_d^{\otimes N}$, can be decomposed into sectors, $\cH^{(\lambda)}$, that allow for the possibility of noiseless encoding and decoding of information.  Using Eq.~\eqref{3} the sectors $\cH^{(\lambda)}$ are the Hilbert spaces arising from the composition of two {\em virtual} quantum systems with corresponding state spaces $\cM^{(\lambda)},\,\cN^{(\lambda)}$~\cite{Z01}.  As the collective noise of the channel acts only on the subsystem associated with $\cM^{(\lambda)}$ this subsystem is a \defn{decoherence-full} subsystem. On the other hand, the subsystem associated with $\cN^{(\lambda)}$ is a \defn{decoherence-free}, or noiseless, subsystem.   The sector $\cH^{(\lambda)}$ is a decoherence-free \defn{subspace}~\cite{ZR98} if for any state $\ket{\psi^{(\lambda)}}\in\cH^{(\lambda)}$,
\begin{equation}
U^{(\lambda)}_{g_i}\ket{\psi^{(\lambda)}}=u^{(\lambda)}(g_i)\ket{\psi^{(\lambda)}}, \quad \forall g_i\in G,
\label{7}
\end{equation}
where $u^{(\lambda)}(g_i)\in\C$ with $\lvert u^{(\lambda)}(g_i)\rvert=1$.  It follows that $\cH^{(\lambda)}$ is a decoherence-free subspace if and only if the dimension of the decoherence-full subsystem, $\cM^{(\lambda)}$, is trivial.  In this case one makes use of the entire subspace, $\cH^{(\lambda)}$, to store and transmit quantum information.  If the decoherence-full subsystems are of non-trivial dimension then according to Eq.~\eqref{5} no state $\ket{\psi^{(\lambda)}}\in\cH^{(\lambda)}$ is invariant under collective noise.  However, the states $\ket{\beta}\in\cN^{(\lambda)}$, associated with the decoherence-free subsystem, are unaffected by the noise of the channel. In this case one makes use of a decoherence-free \defn{subsystem} to store and transmit quantum information. Henceforth, we abbreviate both decoherence-free subspaces and subsystems as DFS.

\subsection*{Example: $\mathrm{SU}(2)$}

We illustrate the use of a DFS for the most general type of collective noise on $N$, two-dimensional quantum systems that is associated with $\mathrm{SU}(2)$ (see appendix~\ref{appendsu2} for more details).  The state $\ket{j_1,m_1}\otimes\cdots\otimes\ket{j_N,m_N}$ can be written, using the familiar Clebsh-Gordan decomposition, as a linear superposition of the orthonormal basis states $\{\ket{J,M,\beta}\}$, where $J$ labels the total angular momentum, $M$ labels the projection of total angular momentum onto the $z$-axis, and $\beta$ labels the various ways $N$ spins can add to a particular total angular momentum $J$.  As we explain in appendix~\ref{appendsu2}, the smallest non-trivial DFS occurs for the case of three qubits.  Using the logical basis 
$
\ket{0_L}\equiv c_1\ket{J=\frac{1}{2},M=\frac{1}{2},\beta=0}+c_2\ket{J=\frac{1}{2},M=\frac{-1}{2},\beta=0},\\
\ket{1_L}\equiv d_1\ket{J=\frac{1}{2},M=\frac{1}{2},\beta=1}+d_2\ket{J=\frac{1}{2},M=\frac{-1}{2},\beta=1},
$
where $|c_1|^2+|c_2|^2=1$ and $|d_1|^2+|d_2|^2=1$, one logical qubit can be transmitted noiselessly through the channel. It follows that the rate of transmission of quantum information is $1/3$, i.e.~three physical qubits are required in order to transmit one logical qubit.

In the general case where $N$, $d$-dimensional systems are used several DFS exist.  However, as the action of collective noise of the channel destroys coherences between different irrep sectors, $\cH^{(J)}$, of the total Hilbert space only a single DFS can be used to transmit quantum data.  Hence, to achieve the maximum possible transmission rate the sender and receiver must utilize the DFS with the largest dimension. In the limit where $N\to\infty$ it is known that the rate of transmission of quantum data using the largest available DFS is $1-\mathcal{O}\left(\log_2(N)\right)/N$~\cite{KBLW01}.

To encode the logical qubit one needs to be able to perform the transformation that maps the tensor product basis, $\ket{j_1,m_1}\otimes\cdots\otimes\ket{j_N,m_N}$, to the block diagonal basis, $\ket{J,M,\beta}$. This basis transformation is accomplished by the Schur transform.  The latter is the matrix whose columns are $\ket{J,M,\beta}$.  It was also shown that the Schur transform, for encoding and decoding information in a DFS, can be efficiently implemented, up to an arbitrary error $\epsilon$, using a number of elementary gates that grows as $N\cdot\mathrm{poly}(\log(N),d,\log(\epsilon^{-1}))$, where $N$ is the total number of systems, and $d$ is their dimension~\cite{BCH06}. 

This is achieved by first performing the Clebsch-Gordan transformation described above on the first three systems, whose output is a linear superposition between the states of a $J=1/2$ system and a $J=3/2$ system.  A second application of the Clebsch-Gordan transformation, between two systems whose joint state is in a linear superposition of a $J=1/2$ system and a $J=3/2$ system, and the forth qubit results in a linear superposition between states of systems with $J=0,1,2$. Continuing this way, the final Clebsch-Gordan transformation couples the $N^{\mathrm{th}}$ qubit with two systems whose joint state is in a linear superposition of $0\leq J\leq \left(\frac{N-1}{2}\right)$. This results in a cascade of $N$ applications of the Clebsch-Gordan transformation mapping the uncoupled basis, $\ket{j_1,m_1}\otimes\ket{j_2,m_2}\otimes\ldots\otimes\ket{j_N,m_N}$, to the total angular momentum basis, $\ket{J,M,\beta}$.  The resulting quantum circuit implementing the Schur transform requires a total of $N\cdot\mathrm{poly}(\log(N),d,\log(\epsilon^{-1}))$ elementary gates in order to be implemented (see~\cite{BCH06b} for more details).

Before we introduce our protocol for transmitting quantum data through channels subject to collective noise, we will need a few useful tools from representation theory (see~\cite{Sternberg} for a good exposition on the subject), which we review here briefly.  Two elements, $g_i,g_k\in G$, are said to be \defn{conjugate} if there exists an element, $g_l\in G$, such that $g_l\cdot g_i\cdot g_l^{-1}=g_k$.   One can define an equivalence relation on the group $G$: each element within the set $[g_i]\equiv \{g_k\in G\;\mathrm{for\, which}\;\exists\, g_l\in G\;\mathrm{such\,that}\;g_k=g_l\cdot g_i\cdot g_l^{-1}\}$ is equivalent to $g_i$.  The set $[g_i]$ is called the \defn{conjugacy class} of $g_i$.  Furthermore, every element of $G$ belongs to one and only one conjugacy class.  Therefore, the conjugacy classes, $\{[g_1],\ldots,[g_s]\}$, form a partition of $G$. Here and in the following $s$ denotes the number of conjugacy classes. An important result from representation theory is that the number of inequivalent irreps of any finite or compact Lie group, $G$, is equal to the number of conjugacy classes of the group.

Another important concept we will use frequently is that of the character of a representation.  As we are concerned with channels associated with finite groups it suffices to consider only unitary representations, since every representation of a finite group is equivalent to a unitary representation~\cite{Sternberg}.  The latter are homomorphisms between $G$ and $\mathbb{U}(n)$, the group of $n\times n$ unitary matrices. It follows that the representation of $g_l\cdot g_i\cdot g_l^{-1}$ is of the form $U_{g_l}U_{g_i}U_{g_l}^{-1}$. The \defn{character} of $U_{g_i}$ is defined as $\chi_{g_i}\equiv\mathrm{tr}(U_{g_i})$. Since $\mathrm{tr}(U_{g_l}U_{g_i}U_{g_l}^{-1})=\mathrm{tr}(U_{g_i})$, all elements in the same conjugacy class, $[g_i]$, have the same character, which we denote by $\chi_{[g_i]}$. As there are $s$ conjugacy classes, $\{[g_1],\ldots,[g_s]\}$, the \defn{compound character}, $\chi$, of a representation, $U$, is an $s$-dimensional vector whose entries, $\chi_i$, are $\chi_{[g_i]}$.  The compound character of an irrep, $U^{(\lambda)}$, is denoted by $\chi^{(\lambda)}$ and is again an $s$-dimensional vector whose entries, $\chi^{(\lambda)}_i$, are $\chi^{(\lambda)}_{[g_i]}$.  The \defn{character table} for a finite group, $G$, is a convenient way to display the compound characters of all irreducible representations, $U^{(\lambda)}$ of $G$, where each row, indexed by the irrep label $\lambda$, in the character table contains the characters $\chi^{(\lambda)}_{[g_i]}$ for $i\in (1,\ldots, s)$. It is known that the characters of the irreps of a finite group satisfy the following orthogonality relation
\begin{equation}
\frac{1}{|G|}\sum_{i=1}^s \lvert[g_i]\rvert\chi^{(\lambda)}_i\chi^{(\lambda^\prime)*}_i=\delta_{\lambda,\lambda^\prime},
\label{13}
\end{equation}
where $\lvert[g_i]\rvert$ denotes the number of elements belonging to the conjugacy class $[g_i]$, and $\chi^{(\lambda^\prime)*}_i$ denotes the complex conjugate of $\chi^{(\lambda^\prime)}_i$.

Knowing $\chi$, and all $\chi^{(\lambda)}$ for a representation, $U$, of a finite group $G$ is sufficient to decompose $U$ into a direct sum of irreps as we now explain. Let $U$ be a representation of a group $G$ that has $s$ conjugacy classes.  Then by the linearity of the trace the compound character, $\chi$, of $U$ is given by
\begin{equation}
\chi=\sum_{\lambda=1}^{s} \alpha^{(\lambda)}\chi^{(\lambda)},
\label{14}
\end{equation}
where $\alpha^{(\lambda)}$ are positive integers corresponding to the multiplicities of the irreps $U^{(\lambda)}$.  Knowing the compound characters, $\chi^{(\lambda)}$, for all the irreps of a group $G$ and $|[g_i]|$ allows to determine all $\alpha^{(\lambda)}$'s in Eq.~\eqref{14} using Eq.~\eqref{13}.

One particular representation of importance in this work is the \defn{regular representation}, $\cR$, of a finite group $G$. Let $\cH$ be a $\lvert G\rvert$-dimensional Hilbert space and associate to every element, $g_i\in G$, one computational basis vector of $\cH$, which we denote by $\ket{g_i}$. In this basis, $\cR_{g_k}$, for any $ g_k\in G$, is a $\lvert G\rvert\times\lvert G\rvert$ permutation matrix which maps the set $\{\ket{g_i}\}_{g_i\in G}$ into itself. More precisely, we have
\begin{equation}
\cR_{g_k}\ket{g_i}=\ket{g_k\cdot g_i}\in\{\ket{g_l}\}_{g_l\in G},\forall g_k, g_i\in G,
\label{15}
\end{equation}
where $g_k\cdot g_i$ is the group product between the elements $g_k,\,g_i\in G$.  As the regular representation acts by permuting the computational basis vectors amongst themselves it follows that
\begin{equation}
\mathrm{tr}(\cR(g_i))=\left\{\begin{array}{l l}
                      |G| & \quad \mbox{if $g_i=e$}\\
                      0 & \quad \mbox{otherwise.}\\ \end{array} \right.
\label{16}
\end{equation}
Thus, the character of the regular representation is a vector whose first entry is $|G|$ and the rest are zero.  It can be easily shown, using Eq (\ref{13}) and the fact that $|[e]|=1$, that the regular representation contains every irrep, $U^{(\lambda)}$, a number of times equal to the dimension, $d_\lambda$, of $U^{(\lambda)}$.  That is
\begin{equation}
\cR_{g_i}=\bigoplus_\lambda d_\lambda U^{(\lambda)}_{g_i}.
\label{17}
\end{equation}
Computing the characters on both sides of Eq.~\eqref{17} gives a useful relation between the order of the group and the dimensions of the irreps:
\begin{align}\nonumber
\chi(\cR_e)=|G|&=\sum_\lambda d_\lambda \chi^{(\lambda)}_e\\
&=\sum_\lambda d_\lambda^2.
\label{18}
\end{align}

\section{\label{sec:3}The protocol}
In this section we describe our communication protocol where we consider collective noise described by  an arbitrary discrete group, $G$, with a number of elements equal to the order, $\lvert G\rvert$, of the group (Sec.~\ref{theory}). We show how to encode and decode $m$ logical qudits using $m+r$ physical qudits prepared in a DFS. We use $r$ auxiliary qudits to construct a set of perfectly distinguishable states in such a way that the collective noise maps these states into each other. In order to encode $m$ logical qudits of quantum information an entangled state is prepared, between the $r$ auxiliary qudits and the $m$ logical qudits, that remains invariant under collective noise. The decoding procedure consists of a unitary correction on the $m$ logical qudits which is conditioned on the outcome of a projective measurement on the first $r$ auxiliary qudits.  In Sec.~\ref{Examples}, we show how our protocol is implemented for a few examples of collective noise channels.  

\subsection{\label{theory}Encoding and decoding of quantum data.}

We begin by considering a quantum channel whose noise is described by a unitary representation, $U$, of a finite group, $G$, acting on the Hilbert space, $\cH_d$.  We say that the representation, $U$, of a finite group, $G$, is isomorphic to $G$ if there is a one-to-one correspondence between the elements, $g_i\in G$, and the matrix representation of these elements, $U_{g_i}$.

For an $r$-qudit state $\ket{\psi}\in\cH_d^{\otimes r}$ and a representation $U$, we define the set of states $S^{(U)}_{(r,\ket{\psi})}\equiv\{\ket{\psi(g_i)}=U_{g_i}^{\otimes r}\ket{\psi},\,g_i\in G\}$. We will be interested in those sets $S^{(U)}_{(r,\ket{\psi})}$ for which the following two conditions are fulfilled. For any pair $g_i,\,g_k\in G$, it holds that:
\begin{enumerate}
\item[(1)] the states $\ket{\psi(g_i)},\,\ket{\psi(g_k)}\in S^{(U)}_{(r,\ket{\psi})}$ are mutually orthogonal, i.e.~$\bracket{\psi(g_i)}{\psi(g_k)}=\delta_{ik}$,
\item[(2)] $U_{g_k}^{\otimes r}\ket{\psi(g_i)}=\ket{\psi(g_k\cdot g_i)}\in S^{(U)}_{(r,\ket{\psi})}$, where $g_k\cdot g_i$ denotes the group product between $g_i,g_k\in G$.  In particular, the set $S^{(U)}_{(r,\ket{\psi})}$ is closed under the action of $U^{\otimes r}$.
\end{enumerate}
We refer to a set, $S^{(U)}_{(r,\ket{\psi})}$, fulfilling both conditions as the set of {\em token states}.
We now show how the existence of a set of token states allows the construction of a protocol for error-free transmission of quantum information through a noisy channel. 

Let the noise of the channel be described by the set of operators $\{U_{g_i},\,g_i\in G\}$. Let us assume that there exists an integer, $r$, and a state, $\ket{\psi}\in\cH_d^{\otimes r}$, such that $S^{(U)}_{(r,\ket{\psi})}$ fulfills condition $(1)$ and $(2)$. The sender, Alice, wishes to transmit a state, $\ket{\phi}\in\cH_d^{\otimes m}$, to the receiver, Bob. To that end Alice prepares $m+r$ qudits in the state
\begin{equation}
\ket{\chi_{\phi}}=\frac{1}{\sqrt{|G|}}\sum_{g_i\in G} \ket{\psi(g_i)}\otimes (U_{g_i}^{\otimes m}\ket{\phi}),
\label{20}
\end{equation}
with $\ket{\psi(g_i)}\in S^{(U)}_{(r,\ket{\psi})}$.
Sending the $r+m$ qudits prepared in the state of Eq.~\eqref{20} through the channel yields
\begin{eqnarray}\nonumber
\ket{\chi'_{\phi}}&=&U_{g_k}^{\otimes (r+m)}\ket{\chi_{\phi}} \\   \nonumber
&=&\frac{1}{\sqrt{|G|}}\sum_{g_i\in G} U_{g_k}^{\otimes r}\ket{\psi(g_i)}\otimes (U_{g_k}U_{g_i})^{\otimes m}\ket{\phi}),\\
\label{21}
\end{eqnarray}
for some $g_k\in G$. Since $U^{\otimes r}_{g_k}\ket{\psi(g_i)}=\ket{\psi(g_k\cdot g_i)}$ (condition (2)), and $U_{g_k}U_{g_i}=U_{g_k\cdot g_i}=U_{g_l}$, for some $l\in \{0,\ldots, |G|-1\}$, Eq.~\eqref{21} yields
\begin{align}\nonumber
U_{g_k}^{\otimes (r+m)}\ket{\chi_{\phi}}&=\frac{1}{\sqrt{|G|}}\sum_{g_l\in G} \ket{\psi(g_l)}\otimes (U_{g_l}^{\otimes m}\ket{\phi})\\
&=\ket{\chi_\phi}.
\label{22}
\end{align}
Thus, Bob receives the same state $\ket{\chi_{\phi}}$. Therefore, the state of Eq.~\eqref{20} lies in a DFS as it is invariant under the action $U_{g_k}^{\otimes(m+r)}$ for all $g_k\in G$. As we will show below, the number of auxiliary systems, $r$, required for the construction of the token states is independent of the number of logical qudits, $m$.

To decode the quantum data Bob performs the measurement described by $\{A_i=\ketbra{\psi(g_i)}\, \mathrm{for}\; i\in (0,\ldots, \lvert G\rvert-1), \,A_\perp=I-\sum_{i=0}^{\lvert G\rvert-1} A_i\}$ on the first $r$ qudits and obtains outcome $i\in (0,\ldots,\lvert G\rvert-1)$. Note that outcome  $A_\perp$ has zero probability of occurring. Conditioned on the outcome, $i$, the unitary $U_{g_i^{-1}}^{\otimes m}$ is performed on the remaining $m$ qudits to retrieve $\ket{\phi}\in\cH_d^{\otimes m}$. Unlike error-correcting codes, no information about the error is obtained, only about which unitary transformation is needed in order to retrieve the message.

We also note that our protocol is naturally suited to a communication scenario involving multiple receivers.  Suppose Alice wishes to distribute the state $\ket{\phi}\in\cH_d^{\otimes m}$ to multiple parties through communication channels whose noise is described by $\{U_{g_i},\, g_i\in G\}$, where the same operation, $U_{g_i}$, is applied on all transmitted qudits. Alice prepares the state in Eq.~\eqref{20}, and sends the first $r$ qudits to a single receiver while distributing the remaining $m$ qudits amongst multiple receivers.  The party holding the first $r$ qudits performs the collective measurement described above, and communicates the outcome to the remaining parties who can each apply the appropriate correction locally on their individual systems.

Alternatively, Alice and Bob can perform a measure and re-align protocol~\cite{BRST09}. Alice prepares the state $\ket{\chi_\phi}=\ket{\psi(g_i)}\otimes U_{g_i}^{\otimes m}\ket{\phi}$, for some $g_i\in G$, and sends this state through the channel to Bob. Bob's description of the state sent to him by Alice is given by the ensemble of states $\{p_k,\, \ket{\psi(g_k)}\otimes U_{g_k}^{\otimes m}\ket{\phi}\}$, where $g_k\in G$.  By performing the same measurement described above on the first $r$ qudits, Bob obtains the outcome $i$ and performs the correction, $U_{g_i^{-1}}$, on the remaining $m$ qudits to retrieve the state $\ket{\phi}$ with certainty. Since two parties sharing a collective noise channel, associated with a finite group, $G$, is equivalent to two parties sharing a perfect quantum channel but lacking a shared frame of reference associated with the finite group $G$, the protocol described in this paragraph is equivalent to a reference frame alignment protocol~\cite{AJV01, BBM04b, CDPS04a}.

Indeed, it is known that for reference frame alignment protocols associated with finite groups there exist states, prepared by Alice, and measurements, performed by Bob, that allow the two parties to perfectly align their respective frames of reference~\cite{SG12}. However, the sender and receiver might wish for their respective reference frames to remain hidden, as a party's reference frame might serve as a way of identifying themselves to some third party.  Furthermore, there are instances where it is advantageous to maintain the noise of a quantum channel, as noisy channels can improve the performance of certain quantum information primitives such as bit commitment~\cite{SR01,HOT06}.  For these cases, it is beneficial for Alice and Bob to utilize a communication protocol that does not allow either party to learn about the other party's reference frame (or equivalently about the action of the channel).

In this paper our goal is to construct an error-avoiding protocol whose implementation, in terms of elementary gates, is more efficient than the best currently known protocols utilizing a DFS. Therefore, we will focus mainly on the DFS protocol outlined above (Eq.~\eqref{20}). In determining the number of gates required to implement the DFS protocol (Sec.~\ref{sec:4}) we will also determine an upper bound on the number of gates required to perform the measure and re-align protocol.

We now show how the set of token states, $S^{(U)}_{(r,\ket{\psi})}$, can be constructed using the following theorem, which we prove in Appendix~\ref{append1}. 
\newline
\begin{theorem}
Let $U$ be a representation of some finite group, $G$, on a Hilbert space, $\cH_d$, that is isomorphic to $G$.  Then there exists an integer, $r$, and a state $\ket{\psi}\in\cH_d^{\otimes r}$, such that $S^{(U)}_{(r,\ket{\psi})}$ satisfies both conditions (1) and (2), if and only if $U^{\otimes r}$ contains the regular representation, $\cR$, of $G$.
\label{thm:1}
\end{theorem}
Using Theorem~\ref{thm:1}, the state $\ket{\psi}\in\cH_d^{\otimes r}$ can then be chosen to be
\begin{equation}
\ket{\psi}=\sum_\lambda\sqrt{\frac{d_\lambda}{\lvert G\rvert}}\sum_{n=1}^{d_\lambda}\ket{\xi^{(\lambda)}_n}\otimes\ket{\zeta^{(\lambda)}_n},
\label{19}
\end{equation}
where the sum is taken over all irreps. The set $\left\lbrace\ket{\xi^{(\lambda)}_n}\right\rbrace$ $\left(\left\lbrace\ket{\zeta^{(\lambda)}_n}\right\rbrace\right)$ denotes an orthonormal basis of $\cM^{(\lambda)}$ $\left(\cN^{(\lambda)}\right)$, and $d_\lambda=\mathrm{dim}(\cM^{(\lambda)})$.  We note that the state of Eq.~\eqref{19} was shown to optimize the maximum likelihood of a correct guess in a reference frame alignment protocol~\cite{CDPSM04b,CDPS06}.  Recall that $\cM^{(\lambda)}$ denotes the space on which the irrep, $U^{(\lambda)}$, of $G$ acts non-trivially, whereas $\cN^{(\lambda)}$ denotes the space on which $U^{\otimes r}$ acts trivially. Note that, since $U^{\otimes r}$ contains the regular representation, the dimension of the multiplicity space, $\cN^{(\lambda)}$, is at least $d_\lambda$. We will show in Appendix~\ref{append1} that it is always possible to choose $r$ sufficiently large (and independent of $m$) such that there exists a state $\ket{\psi}\in\cH_d^{\otimes r}$ so that $S^{(U)}_{(r,\ket{\psi})}$ fulfills condition (1) and (2).

\subsection{Examples \label{Examples}}
In this subsection we explicitly construct the set of token states, $S^{(U)}_{(r,\ket{\psi})}$, and the state $\ket{\chi_\phi}$ of Eq.~\eqref{20}, for the case of the Pauli channel, and for the cyclic groups $\mathbb{Z}_3$ and $\mathbb{Z}_N$.

\subsubsection{The Pauli channel}
Suppose that Alice and Bob share a \defn{Pauli} channel, i.e.~a channel whose noise is described by the set of operators $\{e,x,y,z\}\equiv \{I,\sx,i\sy, \sz\}$. Using the protocol above, we need to construct a set of four token states, one for each element of the Pauli group, that are orthonormal and closed under the action of the Pauli operators.  Thus, we require that $r\geq 2$.   Infact, $r=2$ suffices as we now prove.  The set of tensor products of Pauli operators, $\{I^{\otimes 2}, \sx^{\otimes 2},(i\sy)^{\otimes 2},\sz^{\otimes 2}\}$, forms a four-dimensional representation of the Klein group $K_4$---the smallest, non-cyclic, finite abelian group. The character table for the inequivalent irreps, $U^{(\lambda)}$, of $K_4$ is shown in Table~\ref{tbl1}.
\begin{table}[htb]
\caption{The character table for $K_4$}
\label{tbl1}
\begin{tabular}{c|c|c|c|c}
$\chi$&$[g_0$=e]&$[g_1$=x]&$[g_2$=y]&$[g_3$=z]\\
\hline
$U^{(0)}$&1&1&1&1\\
\hline
$U^{(1)}$&1&1&-1&-1\\
\hline
$U^{(2)}$&1&-1&1&-1\\
\hline
$U^{(3)}$&1&-1&-1&1
\end{tabular}
\end{table}
As all the irreps of the Klein group are one-dimensional~\cite{Sternberg}, computing the compound character, $\chi$, of $\{U_{g_i}^{\otimes 2},\, g_i\in K_4\}$, one finds that $\chi=(4,0,0,0)$.  As the latter is the character of the regular representation, $\cR=U^{(0)}\oplus U^{(1)}\oplus U^{(2)}\oplus U^{(3)}$, of $K_4$ given by
\begin{align}\nonumber
\cR_{g_0}=\left(\begin{matrix} 1&0&0&0\\0&1&0&0\\0&0&1&0\\0&0&0&1\end{matrix}\right),\, \cR_{g_1}=\left(\begin{matrix} 1&0&0&0\\0&1&0&0\\0&0&-1&0\\0&0&0&-1\end{matrix}\right),\\
\cR_{g_2}=\left(\begin{matrix} 1&0&0&0\\0&-1&0&0\\0&0&1&0\\0&0&0&-1\end{matrix}\right),\,
\cR_{g_3}=\left(\begin{matrix} 1&0&0&0\\0&-1&0&0\\0&0&-1&0\\0&0&0&1\end{matrix}\right),
\end{align}
it follows that the collective action of the Pauli operators on two qubits is equivalent to the regular representation of $K_4$.

We now construct the fiducial state, $\ket{\psi}\in\cH_2^{\otimes 2}$, of Eq.~\eqref{19}. Consider the action of $U_{g_i}^{\otimes 2}$ on the state $\ket{\Phi^+}=(\ket{00}+\ket{11})/\sqrt{2}$, where $\ket{ab}\equiv\ket{a}\otimes\ket{b}$.  For any $g_i\in K_4$ it holds that $U_{g_i}^{\otimes 2}\ket{\Phi^+}=\ket{\Phi^+}$. Hence, the state $\ket{\Phi^+}$ belongs to the space $\cH^{(0)}$ on which the trivial representation of $K_4$ acts.  Similarly, it can be shown that $\{\ket{\Psi^+}=(\ket{01}+\ket{10})/\sqrt{2}, \,  \ket{\Psi^-}=(\ket{01}-\ket{10})/\sqrt{2},\, \ket{\Phi^-}=(\ket{00}-\ket{11})/\sqrt{2}\}$ belong in the spaces $\cH^{(1)}, \cH^{(2)}, \cH^{(3)}$ respectively.
Thus, using Eq.~\eqref{19}, $\ket{\psi}\in\cH_2^{\otimes 2}$ is given by
\begin{equation}
\ket{\psi}=\frac{\ket{\Phi^+}+\ket{\Psi^+}+\ket{\Psi^-}+\ket{\Phi^-}}{2}=\ket{0+},
\label{23}
\end{equation}
where $\ket{\pm}=(\ket{0}\pm\ket{1})/\sqrt{2}$. It follows that the set of token states is given by $S^{(U)}_{(2,\ket{0+})}\equiv\{\ket{0+}, \ket{1+},-\ket{1-}, \ket{0-}\}$, corresponding to $\{\ket{\psi(g_0)},\ket{\psi(g_1)}, \ket{\psi(g_2)},\ket{\psi(g_3)}\}$ respectively~\footnote{The negative sign in front of the state $\ket{1-}$ is needed to ensure that the state $\ket{\chi_\phi}$ remains in a DFS.}. Notice that the set of token states is perfectly distinguishable and closed under the action of $\{U_{g_i}^{\otimes 2},\, g_i\in K_4\}$.

Alice wants to send a quantum state, $\ket{\phi}\in \cH_2^{\otimes m}$, to Bob. Following the protocol in Sec.~\ref{theory}, Alice encodes her $m$-partite quantum state by preparing
\begin{align}\nonumber
\ket{\chi_{\phi}}&=\frac{1}{2}\sum_{g_i\in K_4}\ket{\psi(g_i)}\otimes U_{g_i}^{\otimes m}\ket{\phi}\\ \nonumber
&=\frac{1}{2}\left(\ket{0+}\otimes \ket{\phi}+\ket{1+}\otimes\sx^{\otimes m}\ket{\phi}\right.\\
&\left.-\ket{1-}\otimes (i\sy)^{\otimes m}\ket{\phi}+\ket{0-}\otimes\sz^{\otimes m}\ket{\phi}\right).
\label{24}
\end{align}
The state in Eq.~\eqref{24} is invariant under the action  of $\sx^{\otimes(2+m)}$ as
\begin{align}\nonumber
\sx^{\otimes (2+m)}\ket{\chi_\phi}&=\frac{1}{2}\left(\sx^{\otimes 2}\ket{0+}\otimes (\sx\cdot I)^{\otimes m}\ket{\phi}\right.\\ \nonumber
&\left.+\sx^{\otimes 2}\ket{1+}\otimes (\sx\cdot \sx)^{\otimes m}\ket{\phi}\right.\\ \nonumber
&\left.-\sx^{\otimes 2}\ket{1-}\otimes (\sx\cdot i\sy)^{\otimes m}\ket{\phi}\right.\\
&\left.+\sx^{\otimes 2}\ket{0-}\otimes (\sx\cdot \sz)^{\otimes m}\ket{\phi}\right)
\label{25}
\end{align}
which gives
\begin{align}\nonumber
\sx^{\otimes (2+m)}\ket{\chi_\phi}&=\frac{1}{2}\left(\ket{1+}\otimes \sx^{\otimes m}\ket{\phi}+\ket{0+}\otimes \ket{\phi}\right.\\ \nonumber
&\left.+\ket{0-}\otimes \sz^{\otimes m}\ket{\phi}-\ket{1-}\otimes (i\sy)^{\otimes m}\ket{\phi}\right)\\
&=\ket{\chi_\phi}.
\label{26}
\end{align}
A similar calculation shows that  the state $\ket{\chi_\phi}$ is invariant under $U_{g_i}^{\otimes 2+m}$ for all $g_i\in K_4$.  Thus, if Alice sends $2+m$ qubits, prepared in the state $\ket{\chi_\phi}$, through the channel, Bob will receive $2+m$ qubits in the state $\ket{\chi_\phi}$.  Note that this is true for any probability distribution $\{p_{g_i},\, g_i\in (e,x,y,z)\}$.

To decode the quantum data Bob simply performs the measurement $\{A_i=\ketbra{\psi(g_i)},\, i\in(0,1,2,3)\}$ on the first two qubits. Upon obtaining outcome $i$, Bob simply applies $U_{g_i^{-1}}^{\otimes m}$ on the remaining $m$ qubits and retrieves the state $\ket{\phi}\in\cH_2^{\otimes m}$. Notice that the measurement corresponds to single qubit measurements in the $z$ and $x$ basis on auxiliary qubit $1$ and $2$ respectively, and that the required correction operations are local.

\subsubsection{Channel associated with the discrete cyclic group $\mathbb{Z}_3$}
For our second example we consider the case where the noise of the channel is given by the two-dimensional representation of $\mathbb{Z}_3$, the cyclic group of three elements,
\begin{equation}
U_{g_i}=\sum_{n=0}^{1}\omega^{ng_i}\ketbra{n}\quad g_i\in(0,1,2),
\label{27}
\end{equation}
where $\omega=e^{i2\pi/3}$.  The character table for $\mathbb{Z}_3$ is given in Table.~\ref{tbl2}.
\begin{table}[htb]
\caption{The character table for $\mathbb{Z}_3$}
\label{tbl2}
\begin{tabular}{c|c|c|c}
$\chi$&$[g_0]$&$[g_1]$&$[g_2]$\\
\hline
$U^{(0)}$&1&1&1\\
\hline
$U^{(1)}$&1&$\omega$&$\omega^2$\\
\hline
$U^{(2)}$&1&$\omega^2$&$\omega$
\end{tabular}
\end{table}
As $\mathbb{Z}_3$ is abelian, all its irreps are one dimensional and are given by $U_{g_i}^{(\lambda)}=\omega^{\lambda g_i}$.  It follows that the regular representation, $\cR$, of $\mathbb{Z}_3$ is given by
\begin{equation}
\cR_{g_i}=U^{(0)}_{g_i}\oplus U^{(1)}_{g_i}\oplus U^{(2)}_{g_i}=\left(\begin{matrix}1&0&0\\0&\omega^{g_i}&0\\0&0&\omega^{2g_i}\end{matrix}\right).
\label{28}
\end{equation}

Using Eq.~\eqref{27},  $\{U_{g_i}^{\otimes 2},\, g_i\in\mathbb{Z}_3\}$ is given by
\begin{equation}
U_{g_i}^{\otimes 2}=\sum_{n_1,n_2=0}^1\omega^{(n_1+n_2)g_i}\ketbra{n_1,n_2}\quad g_i\in(0,1,2),
\label{29}
\end{equation}
where $n_1$ and $n_2$ are added modulo three. Let us denote by $\{\ket{\lambda,\beta}\}_{\beta=1}^{\alpha^{(\lambda)}}$ the set of orthogonal states $\{\ket{n_1,n_2}\}_{(n_1+n_2){\mathrm{mod}\,3}=\lambda}$. Here,
 $\alpha^{(\lambda)}=\binom{2}{\lambda}$ denotes the number of states corresponding to the same $\lambda$. Then Eq.~\eqref{29} can be written as
\begin{align}\nonumber
U_{g_i}^{\otimes 2}&=\sum_{\lambda=0}^2\sum_{\beta=1}^{\alpha^{(\lambda)}}\omega^{\lambda g_i}\ketbra{\lambda,\beta}\\
&=\bigoplus_{\lambda=0}^2 U_{g_i}^{(\lambda)}\otimes I_{\alpha^{(\lambda)}}, \,\forall g_i\in(0,1,2),
\label{30}
\end{align}
where $I_{\alpha^{(\lambda)}}$ is the ${\alpha^{(\lambda)}}$-dimensional identity.  It follows that $\ket{0}\equiv\ket{00}\in\cH^{(0)}$, the subspace upon which $U^{(0)}$ acts, $\ket{2}\equiv\ket{11}\in\cH^{(2)}$, the subspace upon which $U^{(2)}$ acts, and $\{\ket{1,1}\equiv\ket{01},\ket{1,2}\equiv\ket{10}\}\in\cH^{(1)}$, the subspace upon which $U^{(1)}$ acts. Since every inequivalent irrep of $\mathbb{Z}_3$ is present in $\{U_{g_i}^{\otimes 2},\, g_i\in\mathbb{Z}_3\}$ the latter contains the regular representation as a sub-representation.

We now construct the set of token states for such a channel.  The subspace $\cH^{(1)}$, upon which irrep $U^{(1)}$ acts, is a two-dimensional DFS. However, one state from this subspace is sufficient to construct our token states since for all abelian groups $d_\lambda=\mathrm{dim}(\cM^{(\lambda)})=1$ in Eq.~\eqref{19}.  Hence, we can choose $\ket{1,1}=\ket{01}$ as our standard state from the $\cH^{(1)}$ subspace. Thus, the state $\ket{\psi}\in\cH_2^{\otimes 2}$ of  Eq.~\eqref{19} reads
\begin{equation}
\ket{\psi}=\frac{1}{\sqrt{3}}\left(\ket{00}+\ket{01}+\ket{11}\right)
\label{31}
\end{equation}
or, in the $\{\ket{\lambda,\beta}\}$ basis,
\begin{equation}
\ket{\psi}=\frac{1}{\sqrt{3}}\left(\ket{0}+\ket{1,1}+\ket{2}\right).
\label{31a}
\end{equation}
In what follows we explicitly use Eq.~\eqref{31} but of course the same reasoning would hold if one uses Eq.~\eqref{31a} instead. The set of token states is given by $S^{(U)}_{(2,\ket{\psi})}\equiv\{\ket{\psi(g_i)}=\frac{1}{\sqrt{3}}\left(\ket{00}+\omega^{g_i}\ket{01}+\omega^{2g_i}\ket{11}\right), \, g_i\in(0,1,2)\}$.

To communicate an arbitrary $m$-partite state, $\ket{\phi}\in\cH_2^{\otimes m}$, Alice prepares the state
\begin{align}\nonumber
\ket{\chi_{\phi}}&=\frac{1}{\sqrt{3}}\sum_{g_i\in\mathbb{Z}_3}\ket{\psi(g_i)}\otimes U_{g_i}^{\otimes m}\ket{\phi}\\ \nonumber
&=\frac{1}{3}\left((\ket{00}+\ket{01}+\ket{11})\otimes U_{g_0}^{\otimes m}\ket{\phi}\right.\\ \nonumber
&\left.+(\ket{00}+\omega\ket{01}+\omega^2\ket{11})\otimes U_{g_1}^{\otimes m}\ket{\phi}\right.\\
&\left.+(\ket{00}+\omega^2\ket{01}+\omega\ket{11})\otimes U_{g_2}^{\otimes m}\ket{\phi}\right).
\label{32}
\end{align}
Suppose that the channel performs $U_{g_2}$ on all the qubits.  Then
\begin{align}\nonumber
U_{g_2}^{\otimes(2+m)}&\ket{\chi_{\phi}}=\frac{1}{3}\left(U_{g_2}^{\otimes 2}(\ket{00}+\ket{01}+\ket{11})\right.\\ \nonumber
&\left.\otimes (U_{g_2}U_{g_0})^{\otimes m}\ket{\phi}+U_{g_2}^{\otimes 2}(\ket{00}+\omega\ket{01}+\omega^2\ket{11})\right.\\ \nonumber
&\left.\otimes (U_{g_2}U_{g_1})^{\otimes m}\ket{\phi}+U_{g_2}^{\otimes 2}(\ket{00}+\omega^2\ket{01}+\omega\ket{11})\right.\\
&\left.\otimes (U_{g_2}U_{g_2})^{\otimes m}\ket{\phi}\right).
\label{33}
\end{align}
As representations are homomorphisms, $U_{g_k}U_{g_i}=U_{g_{k+i}}$, Eq.~\eqref{33} gives
\begin{align}\nonumber
U_{g_2}^{\otimes(2+m)}\ket{\chi_{\phi}}&=\frac{1}{3}\left((\ket{00}+\omega^2\ket{01}+\omega\ket{11})\otimes U_{g_2}^{\otimes m}\ket{\phi}\right.\\ \nonumber
&\left.+(\ket{00}+\ket{01}+\ket{11})\otimes U_{g_0}^{\otimes m}\ket{\phi}\right.\\ \nonumber
&\left.+(\ket{00}+\omega\ket{01}+\omega^2\ket{11})\otimes U_{g_1}^{\otimes m}\ket{\phi}\right)\\
&=\ket{\chi_\phi}.
\label{34}
\end{align}
A similar calculation shows that the state $\ket{\chi_\phi}$ is invariant for all $g_i\in\mathbb{Z}_3$.  Thus, if Alice sends $2+m$ qubits, prepared in the state $\ket{\chi_\phi}$, through the channel Bob will receive $2+m$ qubits in the state $\ket{\chi_{\phi}}$.  Note that this is true for any probability distribution $\{p_{g_i},\, g_i\in\mathbb{Z}_3\}$.

To decode the state $\ket{\phi}\in\cH_2^{\otimes m}$, Bob performs the measurement $\{A_i=\ketbra{\psi(g_i)},\, \mathrm{for}\; i\in (0,1,2),\, A_\perp=I-\sum_i A_i \}$ on the first two qubits.  Note that in this example $A_\perp=\ketbra{10}$ so that the set of measurements is complete on $\cH_2^{\otimes 2}$, and that the probability of obtaining this measurement outcome is zero.  Upon obtaining outcome $i$, Bob implements $U_{g_i^{-1}}^{\otimes m}$ on the remaining $m$ qubits and retrieves $\ket{\phi}\in\cH_2^{\otimes m}$.

\subsubsection{Channel associated with the discrete cyclic group  $\mathbb{Z}_N$}
The above example can be easily generalized to the case where the channel's noise is associated to the group $\mathbb{Z}_N$, the cyclic group of $N$ elements, with its action on a $d$-dimensional Hilbert space, $\cH_d$, given by
\begin{equation}
U_{g_i}=\sum_{n=0}^{d-1}\omega^{ng_i}\ketbra{n},\quad g_i\in(0,\ldots, N-1),
\label{35}
\end{equation}
where $\omega=e^{i2\pi/N}$. Notice that here $N$ denotes the number of elements of the group. The character table for $\mathbb{Z}_N$ is given in Table.~\ref{tbl3}.
\begin{table}[htb]
\caption{The character table for $\mathbb{Z}_N$}
\label{tbl3}
\begin{tabular}{c|c|c|c|c}
$\chi$&$[g_0]$&$[g_1]$&$\ldots$&$[g_{N-1}]$\\
\hline
$U^{(0)}$&1&1&$\ldots$&1\\
\hline
$U^{(1)}$&1&$\omega$&$\ldots$&$\omega^{N-1}$\\
\hline
$\vdots$&$\vdots$&$\vdots$&$\ddots$&$\vdots$\\
\hline
$U^{(N-1)}$&1&$\omega^{N-1}$&$\ldots$&$\omega^{(N-1)^2}$
\end{tabular}
\end{table}
As $\mathbb{Z}_N$ is abelian all its inequivalent irreps are one-dimensional and are given by $U_{g_i}^{(\lambda)}=\omega^{\lambda g_i}$.  If $N\leq d$, then the representation, $\{U_{g_i},\, g_i\in\mathbb{Z}_N\}$, in Eq.~\eqref{35} contains the regular representation. If $N>d$, then we need to tensor the representation of Eq.~\eqref{35} with itself a number of times, $r$, such that every inequivalent irrep appears at least once. In Appendix~\ref{append1}, we show that $r$ is finite and depends only on the representation, $U$, of $G$. In Appendix~\ref{append2} we show how to explicitly compute $r$ using some examples and here we will derive it for $\mathbb{Z}_N$.

Let $N>d$ and consider the representation $\{U_{g_i}^{\otimes r},\, g_i\in\mathbb{Z}_N\}$. Using Eq.~\eqref{35},  the latter is given by
\begin{equation}
U_{g_i}^{\otimes r}=\sum_{n_1\ldots n_r=0}^{d-1} \omega^{(n_1+\ldots+n_r)g_i}\ketbra{n_1\ldots n_r}
\label{36}
\end{equation}
for all $g_i\in\mathbb{Z}_N$, where $n_1,\ldots, n_r\in(0,\ldots, d-1)$ and $n_1+\ldots+n_r$ are added modulo $N$. Similarly to the $\mathbb{Z}_3$ example above we use the notation $\{\ket{\lambda,\beta}\}_{\beta=0}^{\alpha^{(\lambda)}}\equiv \{\ket{n_1,\ldots,n_r}\}_{(n_1+\ldots+n_r){\mathrm{mod}\,N}=\lambda}$, where
 $\alpha^{(\lambda)}$ denotes the number of states corresponding to the same $\lambda$. Using this notation, the operators $U_{g_i}^{\otimes r}$ may be re-written as
\begin{align}\nonumber
U_{g_i}^{\otimes r}&=\sum_{\lambda=0}^{r(d-1)}\sum_{\beta=1}^{\alpha^{(\lambda)}}\, \omega^{\lambda g_i}\ketbra{\lambda,\beta}\\
&=\bigoplus_{\lambda=0}^{r(d-1)} U_{g_i}^{(\lambda)}\otimes I_{\alpha^{(\lambda)}},\,\forall g_i\in\mathbb{Z}_N.
\label{37}
\end{align}
Since $\mathbb{Z}_N$ has $N$ inequivalent irreps,  it follows from Eq.~\eqref{37} that $\{U_{g_i}^{\otimes r},\,g_i\in\mathbb{Z}_N\}$ contains the regular representation whenever $r(d-1) \geq N-1$.  That is, the space $\cH_{\cR}\equiv\cH_2^{\otimes r'}$, with $r'=\log_2 N$, on which the regular representation acts is embedded in the higher dimensional space  $\cH_d^{\otimes r}$, with $r=\lceil \frac{N-1}{d-1}\rceil$. 
Notice that this corresponds to an exponential increase of required resources: $r=\lceil \frac{N-1}{d-1}\rceil$ qudits are required to ensure that the set $S^{(U)}_{(r,\ket{\psi})}$ containing $r'=\log_2 N$ states satisfies properties (1) and (2).  Despite this exponential overhead we find an efficient implementation for the group $\mathbb{Z}_N$ that scales only linearly with the number of group elements $N$. 

We now construct the set of token states for such a channel.  Notice that the states $\{\ket{\lambda,\beta}\}_{\beta=1}^{\alpha^{(\lambda)}}$, for a given $\lambda$, form an $\alpha^{(\lambda)}$-dimensional DFS.  However, one state from this subspace is sufficient to construct our token states since for all abelian groups $d_\lambda=\mathrm{dim}(\cM^{(\lambda)})=1$ in Eq.~\eqref{19}. Without loss of generality, we choose $\ket{\lambda,1}$ for each irrep $\lambda$.  In the computational basis the state $\ket{\lambda,1}$ corresponds to the $r$-qudit state where the first $\lambda$ qudits are in state $\ket{1}$, and the remaining $r-\lambda$ qudits are in the state $\ket{0}$. Hence,  the fiducial state, $\ket{\psi}\in\cH_d^{\otimes r}$, of Eq.~\eqref{19} is given by
\begin{equation}
\ket{\psi}=\sqrt{\frac{1}{N}}(\ket{0\ldots00}+\ket{0\ldots01}+\ket{0\ldots11}+\ldots+ \ket{1\ldots11}),
\label{38a}
\end{equation}
or in the $\{\ket{\lambda,\beta}\}$ basis
\begin{equation}
\ket{\psi}=\sqrt{\frac{1}{N}}\sum_{\lambda=0}^{N-1} \ket{\lambda,1}.
\label{38b}
\end{equation}
The set of token states, $S^{(U)}_{(r,\ket{\psi})}$ is therefore given by
\bea
\left\lbrace\ket{\psi(g_i)}=\sqrt{\frac{1}{N}}\sum_{\lambda=0}^{N-1}\omega^{\lambda\cdot g_i}\ket{\lambda,1},\, g_i\in\mathbb{Z}_N\right\rbrace,
\label{39}
\eea
where we have chosen the more compact form of Eq.~\eqref{38b} for the state $\ket{\psi}\in\cH_d^{\otimes r}$.
To communicate an arbitrary $m$-partite state, $\ket{\phi}\in\cH_d^{\otimes m}$, Alice prepares the state
\begin{equation}
\ket{\chi_{\phi}}=\frac{1}{N}\sum_{g_i\in\mathbb{Z}_N} \sum_{\lambda=0}^{N-1}\omega^{\lambda\cdot g_i}\ket{\lambda,1}\otimes U_{g_i}^{\otimes m}\ket{\phi}.
\label{40}
\end{equation}
Similarly to the previous example, it can be shown that $\ket{\chi_{\phi}}$ is invariant under $U_{g_i}^{\otimes(r+m)}$ for any $g_i\in\mathbb{Z}_N$. Thus, if Alice sends $r+m$ qudits, prepared in the state $\ket{\chi_\phi}$, through the channel, Bob will receive the $r+m$ qudits in the state $\ket{\chi_{\phi}}$ independently of the probability distribution $\{p_{g_i},\, g_i\in\mathbb{Z}_N\}$.

As before, the decoding of the quantum data is achieved by performing the measurement $\{A_i=\ketbra{\psi(g_i)},\, \mathrm{for}\; i\in (0,\ldots, N-1),\,A_\perp=I-\sum_iA_i\}$.  Conditioned on the outcome, $i$, of this measurement the correction $U_{g_i^{-1}}^{\otimes m}$ to the remaining qudits is applied \footnote{Like before, the outcome $A_\perp$ has zero probability of occurrence.}.

In the next section we explicitly calculate the number of elementary gates required to encode and decode quantum data using the protocol described above.

\section{\label{sec:4} Implementation of our protocol}

In this section we analyze the required resources for encoding and decoding quantum data transmitted through collective noise channels described by an arbitrary discrete group, $G$.  As mentioned before, whereas Alice and Bob can communicate using the measure and re-align protocol of~\cite{BRST09}, such a protocol is undesirbale since it requires Bob to learn the action of the channel. As the noise of the channel can be used to offer security~\cite{BRS04}, or to improve the performance of certain quantum information primitives~\cite{SR01,HOT06}, it is advantageous to utilize an error-avoiding protocol that reveals no information about the noise of the channel.

In the following we show how to implement the protocol of Sec.~\ref{sec:3}.  In particular, we discuss the required number of elementary gates of the encoding and decoding circuit, as well as the logical depth of the circuit. A direct comparison to previously introduced DFS schemes~\cite{BCH06} reveals that our method is more efficient, and also achieves the optimal transmission rate in the asymptotic limit.  We also provide an upper bound on the number of elementary gates needed to implement the measure and re-align protocol.
 
\subsection{Encoding circuit}
For ease of exposition we shall assume throughout that all the physical systems used in the protocol are qubits.  We determine the number of elementary gates for the case of qudits at the end of this section.

Recall that our protocol encodes quantum information contained in an $m$-qubit state, $\ket{\phi}\in\cH_2^{\otimes m}$, by preparing the state of Eq.~\eqref{20}, where $\ket{\psi}\in \cH_2^{\otimes r}$ for finite  $r$, is given by Eq.~\eqref{19}.  Notice that there are $r' = \log_2 |G|$ orthogonal token states, which are, however, encoded into $r \geq r'$ qubits to ensure the proper behavior under $\{U_{g_i}^{\otimes r},\,g_i\in G\}$. First we associate to each group element $g_i\in G$ a computational basis vector, 
\bea 
\ket{g_i}\equiv \ket{i_{r'},\ldots i_{1}}\in \cH_2^{\otimes r'}, 
\eea
where $i=\sum_{k=1}^{r'} 2^{k-1} i_k$. Then we define the unitary operation $T$ in such a way that 
\bea
T\ket{0}^{\otimes r-r'}\ket{g_i}=\ket{\psi(g_i)},
\label{operationT}
\eea
i.e.~a computational basis state $\ket{g_i}$ of $r'$ qubits, embedded into $\cH_2^{\otimes r}$, is transformed to a token state $\ket{\psi(g_i)}$ of $r$ qubits. We can now re-write Eq.~\eqref{20} as
\bea
\ket{\chi_\phi}= \big(T \otimes \one\big) \ket{0}^{\otimes r-r'}\frac{1}{\lvert G\rvert}\sum_{g_i\in G}\ket{g_i}\otimes U_{g_i}^{\otimes m}\ket{\phi}.
\label{43}
\eea
The encoding of quantum information takes place in two steps. One first prepares the $r'+m$ qubit state $\frac{1}{\lvert G\rvert}\sum_{g_i\in G}\ket{g_i}\otimes U_{g_i}^{\otimes m}\ket{\phi}$, followed by the $r$ qubit operation $T$. The second step can be implemented using at most ${\cal O}(2^r)$ elementary gates. The latter is an upper bound on the number of gates required to implement the measure and re-align strategy of~\cite{BRS04}.  In the following we concentrate on the first step, in particular the efficient implementation of the unitary operation, $W$, acting on $r'+m$ qubits, defined as
\bea
W\equiv\sum_{g_i\in G}\proj{g_i} \otimes U_{g_i}^{\otimes m}.
\label{44}
\eea

To present a circuit implementation of the gate in Eq.~\eqref{44} we define the unitary operators
\bea
\label{Wgm}
W_{g_i}&=&(\one-\proj{g_i})\otimes\one +\proj{g_i}\otimes U_{g_i},\\
W_{g_i}^m&=&(\one-\proj{g_i})\otimes\one +\proj{g_i}\otimes U_{g_i}^{\otimes m}.
\eea
The gate $W_{g_i}$ implements a unitary operation $U_{g_i}$ only if the control register is in state $\ket{g_i}$. In case $\{\ket{g_i}\}$ forms a complete orthonormal basis on $\cH_2^{\otimes r'}$ (i.e.~$\lvert G\rvert=2^{r'}$), we have
\begin{equation}
W\equiv\prod_{g_i\in G} W_{g_i}^m,
\end{equation}
and therefore
\begin{equation}
W\left(\ket{+}^{\otimes r'}\otimes\ket{\phi}\right)=\frac{1}{\lvert G\rvert}\sum_{g_i\in G} \ket{g_i}\otimes U_{g_i}^{\otimes m}\ket{\phi},
\label{45}
\end{equation}
where $\ket{+}^{\otimes r'}\equiv\frac{1}{\lvert G\rvert}\sum_{g_i\in G}\ket{g_i}$~\footnote{Note that if $\{\ket{g_i}\}$ does not form a complete orthonormal basis, that is $\lvert G\rvert<2^{r'}$, we have to apply $W$ to a state which is the superposition of $\lvert G\rvert$ computational basis states (which does not coincide with $\ket{+}^{r'}$). Such an input state can be easily generated in the following way. Let $\tilde r<r'$ be such that $2^{\tilde r-1}<|G|<2^{\tilde r}$. Then, $\ket{\Psi}=\otimes_{i=2}^{\tilde r} U_{1i}\ket{+}^{\otimes r'}$, where $U_{1i}=\proj{0}\otimes \one+\proj{1}\otimes U_i$. Choosing $U_i$, which acts on qubit $i$, either as $U_i=\one$, or $U_i\ket{+}=\ket{0}$ allows one to generate any desired superposition of the form $|0\rangle|+\ldots ++\rangle + |1\rangle|+0 \ldots 0+\rangle$  with $\tilde r$ gates. Note that not all terms in this sum have the same weight when we write the state in the computational basis. This can however easily be corrected by preparing the first qubit in a state $\cos\alpha |0\rangle + \sin\alpha|1\rangle$ rather than $|+\rangle$ and choosing $\alpha$ appropriately.}.
Notice that $W$ corresponds to the sequence of controlled-unitary gates, $W_{g_i}^m$, for all possible values of $g_i\in G$.

We will now outline a circuit implementing the gate $W$ of Eq.~\eqref{44}.  First, note that $\ket{g_i}$ is a binary representation of the value $i\leq 2^{r'}$ corresponding to the group element $g_i\in G$. The gate $W_{g_i}^m$, for some fixed $g_i\in G$, can be implemented by applying local unitaries $\sx^{(i_{r'})}\otimes\ldots\otimes\sx^{(i_{1})}$ to the first $r'$ qubits such that
$\sx^{(i_{r'})}\otimes\ldots\otimes\sx^{(i_{1})}\ket{i_{r'}\ldots i_{1}}=\ket{1}^{\otimes r'}$ and then applying the gate
\begin{equation}
V_{g_i}^m \equiv (\one-\proj{1}^{\otimes r'})\otimes \one+ \proj{1}^{\otimes r'}\otimes U_{g_i}^{\otimes m}.
\label{47}
\end{equation}
The latter can in turn be implemented by applying the gate
$V_{g_i}=(\one-\proj{1}^{\otimes r'})\otimes \one+ \proj{1}^{\otimes r'}\otimes U_{g_i}$
$m$ times, where the control qubits remain the same but the target qubit is always a new one. The implementation of the gates $V_{g_i}^m$, $W_{g_i}^m$, and $W$ is shown in  Figs.~(\ref{fig1},~\ref{fig2},~\ref{fig3}) respectively.

\subsection{Resources}

We will now count how many elementary gates are required in order to implement $V_{g_i}$. In \cite{BaBe95} it has been shown that a control gate of the form $\Lambda_{r'} (U)=(\one-\proj{1}^{\otimes r'})\otimes \one+ \proj{1}^{\otimes r'}\otimes U$, where $r'$ denotes the number of control qubits, can be implemented using $40(r'-2)+1$
elementary gates \footnote{The factor $4(r'-2)$ is the number of Toffoli gates required to implement a $r'$-control-X gate, and the factor 10 comes from applying these gates twice (before and after), and using 5 CNOTs per Toffoli gate. The $+1$ is the controlled-U operation (see Lemma 7.11 and Lemma 7.2 in \cite{BaBe95}}.  In order to apply a controlled-$U_{g_i}^{\otimes m}$ gate we simply apply the $m$ controlled-$U_{g_i}$ gates with different target qubits in between the two $r'$-controlled-$\sx$ gates (see Figs.~(\ref{fig1},~\ref{fig2})).

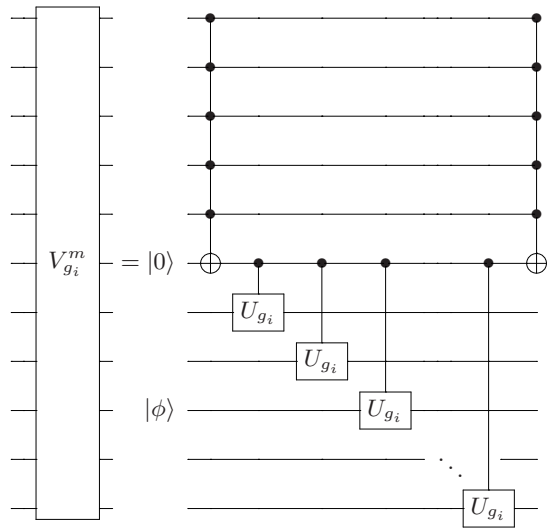
\begin{figure}
\[\Qcircuit @C=.5em @R=0em @!R {
&\qw&\multigate{10}{V_{g_i}^m}&\qw\\
\push{\rule{0em}{2em}}
&\qw&\ghost{V_{g_i}^m}&\qw\\
&\qw&\ghost{V_{g_i}^m}&\qw\\
&\qw&\ghost{V_{g_i}^m}&\qw\\
&\qw&\ghost{V_{g_i}^m}&\qw\\
&\qw&\ghost{V_{g_i}^m}&\qw\\
&\qw&\ghost{V_{g_i}^m}&\qw\\
&\qw&\ghost{V_{g_i}^m}&\qw\\
&\qw&\ghost{V_{g_i}^m}&\qw\\
&\qw&\ghost{V_{g_i}^m}&\qw\\
&\qw&\ghost{V_{g_i}^m}&\qw\\
}\hspace{10mm}
\Qcircuit @C=.5em @R=0em @!R {
&\ctrl{1}&\qw&\qw&\qw&\qw&\qw&\qw&\qw&\ctrl{1}\\
\push{\rule{0em}{2em}}
&\ctrl{1}   &\qw         &\qw&\qw&\qw&\qw&\qw&\qw&\ctrl{1}\\
&\ctrl{1}   &\qw         &\qw&\qw&\qw&\qw&\qw&\qw&\ctrl{1}\\
&\ctrl{1}   &\qw         &\qw&\qw&\qw&\qw&\qw&\qw&\ctrl{1}\\
&\ctrl{1}   &\qw         &\qw&\qw&\qw&\qw&\qw&\qw&\ctrl{1}\\
\lstick{=\ket{0}}&\targ      &\ctrl{1}    &\ctrl{2}&\ctrl{3}&\qw&\qw&\qw&\ctrl{5}&\targ\\
&\qw        &\gate{U_{g_i}} &\qw &\qw     &\qw&\qw&\qw&\qw&\qw\\
&\qw        &\qw          &\gate{U_{g_i}} &\qw&\qw&\qw&\qw&\qw&\qw\\
\lstick{\ket{\phi}}&\qw        &\qw          &\qw &\gate{U_{g_i}}&\qw&\qw&\qw&\qw&\qw\\
&\qw        &\qw          &\qw&\qw&\qw&\push{\rule{0.1em}{0em}}&\ddots&\push{\rule{1em}{0em}}&\qw\\
&\qw         &\qw        &\qw&\qw&\qw&\qw&\qw&\gate{U_{g_i}}&\qw\\
}
\]
\caption{The quantum circuit implementation of the encoding operation $V_{g_i}^m$.  We note that the circuit implementation of the  $r'$-controlled Toffoli gates in this circuit can be found in Lemma 7.2 of~\cite{BaBe95}.}
\label{fig1}
\end{figure}
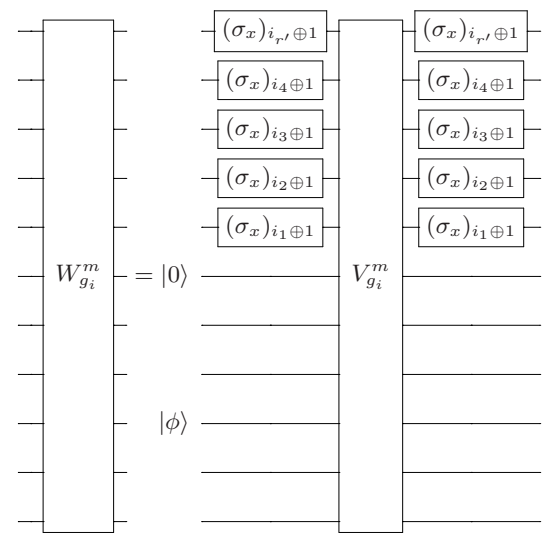
\begin{figure}
\[\Qcircuit @C=.5em @R=0em @!R {
&\qw&\multigate{10}{W_{g_i}^m}&\qw\\
\push{\rule{0em}{2em}}
&\qw&\ghost{W_{g_i}^m}&\qw\\
&\qw&\ghost{W_{g_i}^m}&\qw\\
&\qw&\ghost{W_{g_i}^m}&\qw\\
&\qw&\ghost{W_{g_i}^m}&\qw\\
&\qw&\ghost{W_{g_i}^m}&\qw\\
&\qw&\ghost{W_{g_i}^m}&\qw\\
&\qw&\ghost{W_{g_i}^m}&\qw\\
&\qw&\ghost{W_{g_i}^m}&\qw\\
&\qw&\ghost{W_{g_i}^m}&\qw\\
&\qw&\ghost{W_{g_i}^m}&\qw\\
}\hspace{10mm}
\Qcircuit @C=.5em @R=0em @!R {
&\gate{(\sigma_x)_{i_{r'}\oplus1}}&\multigate{10}{V_{g_i}^m}&\gate{(\sigma_x)_{i_{r'}\oplus1}}&\qw\\
\push{\rule{0em}{2em}}
&\gate{(\sigma_x)_{i_4\oplus1}}&\ghost{V_{g_i}^m}&\gate{(\sigma_x)_{i_4\oplus1}}&\qw\\
&\gate{(\sigma_x)_{i_3\oplus1}}&\ghost{V_{g_i}^m}&\gate{(\sigma_x)_{i_3\oplus1}}&\qw\\
&\gate{(\sigma_x)_{i_2\oplus1}}&\ghost{V_{g_i}^m}&\gate{(\sigma_x)_{i_2\oplus1}}&\qw\\
&\gate{(\sigma_x)_{i_{1}\oplus1}}&\ghost{V_{g_i}^m}&\gate{(\sigma_x)_{i_{1}\oplus1}}&\qw\\
\lstick{=\ket{0}}&\qw&\ghost{V_{g_i}^m}&\qw&\qw\\
&\qw&\ghost{V_{g_i}^m}&\qw&\qw\\
&\qw&\ghost{V_{g_i}^m}&\qw&\qw\\
\lstick{\ket{\phi}}&\qw&\ghost{V_{g_i}^m}&\qw&\qw\\
&\qw&\ghost{V_{g_i}^m}&\qw&\qw\\
&\qw&\ghost{V_{g_i}^m}&\qw&\qw
}
\]
\caption{The quantum circuit for $W_{g_i}^m$ for any state $\ket{g_i}=\ket{i_{r'}\ldots i_{1}}$, a binary representation of the group element $g_i\in G$. The gate $(\sigma_x)_{i_m\oplus1}$ flips the $m^{\mathrm{th}}$ qubit of the input state, if the $m^{\mathrm{th}}$ digit, $i_m$, in the binary representation of $g_i\in G$ is zero. After implementing the gate $V_{g_i}^m$ the bit string is restored to its original value.}
\label{fig2}
\end{figure}
Thus,
\bea
f(r')\equiv 40(r'-2)+m
\label{fr}
\eea
elementary gates are required to implement $V_{g_i}^m$. Therefore, the number of gates required to implement $W_{g_i}^m$ is $M \equiv 40(r'-2)+m+r'$ 
(another $r'$ operations for local basis change in the control register \footnote{Notice that there would actually be $2r'$ operations for the basis change, $r'$ before and $r'$ after the application of $V_{11 \ldots 1}^m$. However, the second set of $r'$ operations only returns the control register to its initial value, which is not necessary for a single gate $W_{g_j}^m$ and can in fact be combined with the first set of $r'$ operation of the subsequent gate $W_{g_l}^m$.}). Since a total of $|G|$ different gates $W_{g_i}^m$ need to be applied to implement $W$ (one for each group element $g_i\in G$ - see Eq.~\eqref{45} and Fig.~\ref{fig3}), we find that the total number of elementary gates required to implement $W$ is given by
\bea
\lvert G\rvert M=\lvert G\rvert(41r'-80+m),
\label{49}
\eea
where $|G|=2^{r'}$. That is, the resources required to encode the quantum data are linear in the number of qubits, $m$, to be transmitted and scale as $\lvert G\rvert \log(\lvert G\rvert)$.  After performing the number of gates in Eq.~\eqref{49} an additional ${\cal O}(2^r)$ gates are required to implement the unitary $T$ that maps the computational basis states $\{\ket{g_i}\}$ to the set of token states $\{\ket{\psi({g_i})}\}$~\footnote{This is an upper bound on the number of gates required to implement $T$.}.
\begin{figure}
\[
\Qcircuit @C=.5em @R=0em @!R {
&\gate{H}&\multigate{10}{W_{00\ldots00}^m}&\multigate{10}{W_{00\ldots01}^m}&\multigate{10}{W_{00\ldots10}^m}&\qw&\ldots\\
\push{\rule{0em}{2em}}
&\gate{H}&\ghost{W_{00\ldots00}^m}&\ghost{W_{00\ldots01}^m}&\ghost{W_{00\ldots10}^m}&\qw&\ldots\\
\lstick{\ket{00\ldots0}}&\gate{H}&\ghost{W_{00\ldots00}^m}&\ghost{W_{00\ldots10}^m}&\ghost{W_{00\ldots10}^m}&\qw&\ldots\\
&\gate{H}&\ghost{W_{00\ldots00}^m}&\ghost{W_{00\ldots01}^m}&\ghost{W_{00\ldots10}^m}&\qw&\ldots\\
&\gate{H}&\ghost{W_{00\ldots00}^m}&\ghost{W_{00\ldots01}^m}&\ghost{W_{00\ldots10}^m}&\qw&\ldots\\
&\qw&\ghost{W_{00\ldots00}^m}&\ghost{W_{00\ldots01}^m}&\ghost{W_{00\ldots10}^m}&\qw&\ldots\\
&\qw&\ghost{W_{00\ldots00}^m}&\ghost{W_{00\ldots01}^m}&\ghost{W_{00\ldots10}^m}&\qw&\ldots\\
&\qw&\ghost{W_{00\ldots00}^m}&\ghost{W_{00\ldots01}^m}&\ghost{W_{00\ldots10}^m}&\qw&\ldots\\
\lstick{\ket{\phi}}&\qw&\ghost{W_{00\ldots00}^m}&\ghost{W_{00\ldots01}^m}&\ghost{W_{00\ldots10}^m}&\qw&\ldots\\
&\qw&\ghost{W_{00\ldots00}^m}&\ghost{W_{00\ldots01}^m}&\ghost{W_{00\ldots10}^m}&\qw&\ldots\\
&\qw&\ghost{W_{00\ldots00}^m}&\ghost{W_{00\ldots01}^m}&\ghost{W_{00\ldots10}^m}&\qw&\ldots}\hspace{5mm}
\Qcircuit @C=.5em @R=0em @!R {
&\qw&\multigate{10}{W_{11\ldots1}^m}&\qw\\
\push{\rule{0em}{2em}}
&\qw&\ghost{W_{11\ldots1}^m}&\qw\\
&\qw&\ghost{W_{11\ldots1}^m}&\qw\\
&\qw&\ghost{W_{11\ldots1}^m}&\qw\\
&\qw&\ghost{W_{11\ldots1}^m}&\qw\\
&\qw&\ghost{W_{11\ldots1}^m}&\qw\\
&\qw&\ghost{W_{11\ldots1}^m}&\qw\\
&\qw&\ghost{W_{11\ldots1}^m}&\qw\\
&\qw&\ghost{W_{11\ldots1}^m}&\qw\\
&\qw&\ghost{W_{11\ldots1}^m}&\qw\\
&\qw&\ghost{W_{11\ldots1}^m}&\qw}
\]
\caption{The circuit implementation of the encoding circuit $W=\sum_{g_i\in G} W_{g_i}^m$, where $g_i\in G$ is written in binary notation.}
\label{fig3}
\end{figure}
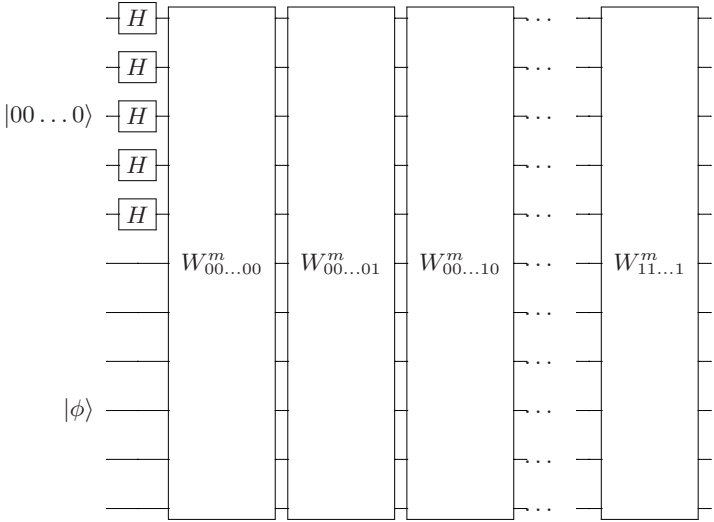

Finally, we note that if the dimension of the representation, $U$, of $G$ is $d$, i.e.~if the operators $\{U_{g_i},\, g_i\in G\}$ act on $d$-dimensional systems, then the number of elementary gates required to implement $W$ only increases by a factor which is independent of $m$ and $\lvert G\rvert$. In this case the unitary transformation, $T$, requires at most ${\cal O}(d^r)$ elementary gates in order to be implemented. In the following, we consider some special groups and show that the required resources can be significantly reduced.

\subsubsection{Abelian groups}
We now discuss a method to implement the gate $W$, given in Eq.~\eqref{44}, for the case of quantum channels whose collective noise is associated with a finite abelian group. Denoting by $g_1,\ldots, g_k$ the generators of the group, then for any $g\in G$ there exists a string, $(l_1(g),\ldots, l_k(g))$, with $l_j(g)\in N$, such that $g=g_1^{l_1(g)}\cdots g_k^{l_k(g)}$. Since we are dealing with finite groups we have that for any $j$ there exists a $L_j$ such that $l_j(g)\leq L_j$ for any $g\in G$. Writing, $\ket{g}=\ket{l_1(g)\ldots l_k(g)}$, where $l_j(g)$ is represented in binary notation, the gate $W$ in Eq.~\eqref{44} becomes
 \bea
W= \sum_{g\in G}\proj{g} \otimes U_{g}^{\otimes m}=\prod_{i=1}^k U_i,
 \label{50}
 \eea
with $U_i=\sum_{l_i=0}^{L_i}\proj{l_i} \otimes (U_{g_i}^{l_i})^{\otimes m}$. Note that each gate $U_i$ is acting on a $L_i$-dimensional control system (i.e.~$\log{L_i}$ control qubits) and $m$ target qubits (the first control system controls how often $g_1$ is applied, the second how often $g_2$ etc.). Since $U_i$ is acting on  $(\log L_i+m)$ qubits and requires at most $L_i$ control gates, the number of elementary gates required to implement $U_i$ is at most $L_i f(\log{L_i})=L_i[40(\log{L_i}-2)+m]$ (see Eq.~(\ref{fr})).  Thus, the total number of required elementary gates is $\sum_{i=1}^k L_i f(\log{L_i}) \leq k\,\mbox{max}_i\{L_i(f(\log{L_i})\}$, which might be substantially smaller than $|G|(41r'-80+m)$ gates required in the general case. Notice that the unitary basis change, $T$, could also be implemented more efficiently (i.e.~with less than ${\cal O}(2^r)$ gates) in certain cases.

\subsubsection{Cyclic groups}

Let us now consider the particular situation where the collective noise of the channel is associated with a general cyclic group, i.e.~a group with only one generating element $h\equiv g_1$. Then, the group elements are given by $h^j$ where $1 \leq j \leq L$ with $L=2^{r'}$. As before, we associate to each group element, $h^j$, the number $j$, which we write in binary notation as $j= j_{r'}\ldots j_{1}$, with $j=\sum_{i=1}^{r'} 2^{i-1} j_i$ and $U_{h^j}=U^j=\prod_{i=1}^{r'} U^{2^{i-1} j_i}$.

The last equation is the key to an efficient implementation of the operation $W$ (Eq. (\ref{44})). Rather than implementing a product of $2^{r'}$ controlled unitary operations $W_{g_i}^m$ (Eq. (\ref{Wgm})), it suffices to perform $r'$ controlled unitary operations that use the $i^{\rm th}$ qubit of the first register as the control, and perform the operation $(U^{2^i j_i})^{\otimes m}$ on the message qubits if the bit value is one. That is, the first control qubit controls whether $\one$ or $U^{\otimes m}$ is applied, the second whether $\one$ or $(U^2)^{\otimes m}$, the third whether $\one$ or $(U^4)^{\otimes m}$ etc. In total, this leads to the implementation of the operation $(U^j)^{\otimes m}$ if the control state is given by $\ket{j} = \ket{j_{r'}\ldots j_{2}j_1}$, corresponding exactly to the operation $W$. As each of these gates consists of $m$ two-qubit gates, we have a total of $m \log{L}=m r'$ gates. This leads to a significant reduction of the required resources for cyclic groups, i.e.~for $\mathbb{Z}_N$, with $N = 2^{r'} \in \mathbb{N}$, we only require $m \log N$ gates.

Notice that for cyclic groups the implementation of the unitary basis change to the token basis, i.e.~the unitary operation, $T$, in Eq.~(\ref{operationT}), can also be done much more efficiently than the upper bound of $\mathcal{O}(d^r)$ operations. We consider $d=2$, i.e.~qubits. First, we notice that for the group $\mathbb{Z}_N$, the required number of qubits to store the token states is given by $r =|G|-1=N-1$ (see Eq.~\eqref{37}), while only $r'= \log_2 N$ qubits are required to label the group elements. The implementation of $T$ then consists of a Schur transformation that maps the computational basis states to the Schur basis states (see Eq.~\eqref{38b}), followed by the Fourier transformation (see Eq.~\eqref{39}).  Notice that the order of the operations can be exchanged and, furthermore, the Fourier transformation just acts on standard basis states of $r' < r$ qubits. Both operations can be implemented efficiently; the Fourier transformation using ${\cal O}(r' \log r')$ gates, and the Schur transformation using ${\cal O}(r\mathrm{poly}(\log r))$ resources (following the results of \cite{BCH06}).

For the group $\mathbb{Z}_N$ the change from the computational basis to the $\ket{\lambda,1}$ basis  (see Eq. (\ref{38b})) can in fact be implemented using only $r+r'$ elementary gates as we now show. As mentioned above, since all irreps of $\mathbb{Z}_N$ are one-dimensional, we simply need to construct one state, $\ket{\lambda,\beta}$, for each $\lambda$ which is a computational basis state containing $\lambda$ ones.  In order to do so, we take $r'$ qubits (the first register) containing the computational basis states $\ket{j_{r'}, j_{r'-1},\ldots, j_{1}}$, and an additional $r=2^{r'}-1$ qubits (the second register) that we partition into $r'$ groups $A_{m}$.   Each group, $A_m$, in the second register corresponds to the $m^{th}$ qubit of the first register and contains $2^{m-1}$ qubits (corresponding to its value in binary representation).

To construct the required states we proceed in two steps.  Firstly, we perform $m$ CNOT operations with the $m^{\rm th}$ qubit in the first register as the control and the $2^{m-1}$ qubits in the group $A_m$ of the second register as targets. The $m$ CNOT operations cause all qubits within the group, $A_m$, in the second register to flip if the $m^{\rm th}$ qubit in the first register is in the state $\ket{j_m} = \ket{1}$, and does nothing to  the qubits in group $A_m$ if $\ket{j_m} = \ket{0}$. Secondly, we apply $r'$ CNOT operations with one of the qubits in $A_m$ of the second register as the control qubit, and the $m^{\rm th}$ qubit of the first register as the target. This ensures that the first register is in the state $\ket{0}^{\otimes r'}$, while the state of the second register contains a total number of ones corresponding to the value of the bit-string $j_{r'} j_{r'-1} \ldots j_{1}$.

For example, the elements of $\mathbb{Z}_N$ can be represented in binary notation using $r'=\log_2 N$ bits.  Without loss of generality we assume that $N$ is an exact power of two~\footnote{If $N$ is not an exact power of two then $r'=\lceil\log_2 N\rceil$.}.  Thus, the first register consists of the $r'$ qubit computational basis states of the form $\ket{j_{r'},j_{r'-1},\ldots,j_2,j_1}$, where $j_i\in (0,1)$ for $i\in (1,\ldots, r')$. The second register consists of $r=2^{r'}-1=N-1$ qubits, initially in the state $\ket{0}$, so that  the initial state of both registers is given by
\begin{equation}
\left(\ket{j_{r'}j_{r'-1}\ldots j_2j_1}\right)\otimes \left(\ket{0}^{\otimes\frac{N}{2}}\ket{0}^{\otimes\frac{N}{4}}\ldots\ket{0}^{\otimes 2}\ket{0}\right),
\label{register1}
\end{equation}
where we have partitioned the $r$ qubits in the second register into $r'$ groups, $A_m$, each containing $2^{m-1}$ qubits. After applying the first $m$ CNOT gates, with the second register as target,  the state of the two registers is
\begin{equation}
\left(\ket{j_{r'}j_{r'-1}\ldots j_2j_1}\right)
\otimes \left(\ket{j_{r'}}^{\otimes\frac{N}{2}}\ket{j_{r'-1}}^{\otimes\frac{N}{4}}\ldots\ket{j_2}^{\otimes 2}\ket{j_1}\right),
\label{register2}
\end{equation}
and after an additional $r'$ CNOT gates, with the first register as target, the final state of the two registers is
\begin{equation}
\left(\ket{0}^{\otimes r'}\right)\otimes\left(\ket{j_{r'}}^{\otimes\frac{N}{2}}\ket{j_{r'-1}}^{\otimes\frac{N}{4}}\ldots\ket{j_2}^{\otimes 2}\ket{j_1}\right).
\label{register3}
\end{equation}

From our discussion above, it follows that the total number of CNOT operations required to implement the basis change is given by $r+r'$. In addition, the Fourier transformation needs to be applied before this basis change, which involves ${\cal O}(r' \log r')$ gates. Thus, the overhead for implementing the operation $T$ (Eq.~(\ref{operationT})) for the case of finite cyclic groups is $r+r'+{\cal O}(r' \log r')=N-1+\log_2 N+{\cal O}(\log_2 N\log(\log_2 N)) = {\cal O}(N)$, i.e.~only {\em linear} with the number of group elements $N$, despite the exponential increase of $r$ compared to $r'$ required for the token states corresponding to the group $\mathbb{Z}_N$.  Together with the efficient implementation of the operation $W$ discussed above we find that for cyclic groups the encoding requires ${\cal O}(m \log N,N)$ elementary gates. 

\subsection{Logical depth}

It should be noted that the logical depth of our protocol is independent of the number of transmitted qubits $m$. This is in contrast to the DFS-based communication scheme put forward in~\cite{BCH06}.  It is easy to see that the logical depth of the circuit to implement $W$ is given by $\lvert G\rvert(41r'-80+1)$, since all the control gates occurring in $V_{g_i}^m$, acting on $m$ qubits originally prepared in $\ket{\phi}$, can be implemented in parallel. Also, the unitary basis change, $T$ in Eq.~(\ref{operationT}), only requires at most ${\cal O}(2^r)$ gates, where $r$ depends only on the representation, $U$, of $G$, leading to a logical depth that is {\em independent} of $m$. The same is true for the more efficient implementations for abelian and cyclic groups discussed above. This allows for a very efficient implementation of our communication scheme.

\subsection{Decoding}

In order to decode the information Bob simply measures in the basis $\{\ket{\psi(g_i)}\}$ and applies, depending on the outcome, one of the operations $U_{g_i^{-1}}^{\otimes m}$ in order to retrieve the state $\ket{\phi}\in\cH_2^{\otimes m}$. In practice, this can be done by implementing the inverse of the unitary operation $T$, appearing in Eq.~\eqref{operationT}, that maps the product basis to the token state basis, i.e.~$T^\dagger$, followed by $r'$ single qubit measurements in the computational basis. Notice that $T^\dagger$ can also be implemented with at most ${\cal O}(2^r)$ basic gates, independent of $m$, or more efficiently for certain groups as shown above. The final correction operation, $U_{g_i^{-1}}^{\otimes m}$, is comprised of single-qubit operations that can be performed independently and in parallel. This makes possible a multi-receiver scenario as described in Sec.~\ref{sec:3}. Notice however, that in general the auxiliary systems need to be transmitted to a single party who also performs the measurement, and then communicates the classical measurement outcome to the different receivers. 

\subsection{Asymptotic transmission rate}

We now compare the asymptotic rate of transmission of quantum information of our protocol to that of a DFS code.  First, let us compute the rate of transmission for the latter.  As not all $\cN^{(\lambda)}$ can be simultaneously utilized,  the maximum number of logical qubits, $m$, that can be transmitted using a number of physical qubits, $n$, is $m=\log_2(\mathrm{max}_\lambda\,\mathrm{dim}(\cN^{(\lambda)}))$.  Hence, the rate of transmission for a DFS code is given by
\begin{equation}
R_{DFS}=\frac{\log_2(\mathrm{max}_\lambda\mathrm{dim}(\cN^{(\lambda)})}{n},
\label{51}
\end{equation}
and it is known that $\lim_{n\rightarrow\infty}R_{DFS}\rightarrow 1-\frac{O(\log n)}{n}$~\cite{KBLW01}.

We now calculate the asymptotic rate of transmission, $R$, using the protocol described in Sec.~\ref{sec:3}.  As for any finite group $G$, $m$ logical qudits can be perfectly transmitted (i.e.~with unit fidelity of transmission) using $r+m$ physical qudits, the rate of transmission is given by
\begin{equation}
R=\frac{m}{r+m}.
\label{52}
\end{equation}
However, as $G$ is a finite group the number of qudits, $r$, required to construct our token states is finite and depends only on the representation, $U$, of $G$ (see Appendix~\ref{append1}).  Hence, in the limit $m\rightarrow\infty$, Eq.~\eqref{52} tends to unity and our protocol achieves the optimal transmission rate.  

\section{Summary and conclusions}\label{Conclusion}

In conclusion, we have introduced a new protocol for transmitting quantum information through channels with collective noise.  We have shown how to transmit $m$ logical qudits using $m+r$ physical qudits at a rate that is optimal in the asymptotic limit. The protocol makes use of ideas both from DFS and error correction. On the one hand, a specific state of system plus ancilla qubits is used that lies in a DFS of the joint system. On the other hand, the ancilla qubits are measured to determine the required correction operation similar to error correction. However, in our protocol no information about the channel and hence the actual error is revealed.

In the case of channels associated with a finite group $G$, the $m$ logical qudits can be transmitted with perfect fidelity, and can be efficiently encoded and decoded. We find that the number of elementary gates required for the encoding and decoding circuit scales as $\mathcal{O}(m,|G|\log|G|, d^r)$, where $r$ is an integer that depends solely on the channel in question, and local measurements.  For the case of finite cyclic groups, $\mathbb{Z}_N$, we discover that the encoding and decoding operations can be efficiently implemented with at most $m\log N+ \mathcal{O}(N)$ operations, where $N$ is the order of the cyclic group.  Moreover, the logical depth of the encoding and decoding circuit for finite groups is independent of the number of logical qudits, $m$.  As the required number of elementary gates scales only linearly in the number of logical qudits, our protocol is more efficient than the best currently known DFS protocols~\cite{BCH06,LNPST11}. Based on our findings, a practical implementation of our protocol seems feasible.

Whereas the implementation of our protocol for finite abelian groups is very efficient, it is not obvious if an efficient implementation of our protocol is feasible for the case of non-abelian groups.  This is due to the fact that, in general, a $\mathcal{O}(d^r)$ overhead is required to implement the unitary operator in Eq.~\eqref{operationT}, which performs a basis change from the computational basis to the token-state basis. Even though we have explicitly shown that for any finite group it is always possible to find such a token state basis, the required overhead depends on the group in question. For a group with $\lvert G\rvert$ elements $r'=\log_2\lvert G\rvert$ qubits are required to label the elements, however $r \geq r'$ qubits are needed to construct a token basis with desired properties. While for the Pauli group we find that $r=r'$, in the case of finite cyclic groups, we discover that $r= (\lvert G\rvert -1)/(d-1)$, i.e.~an exponential overhead. However, despite this exponential increase for cyclic groups, we have shown that the unitary operator in Eq.~\eqref{operationT} can be efficiently implemented using $\mathcal{O}(r)$ operations, i.e.~with an overhead that scales only linear in the number of group elements. Whether such an exponential overhead of $r$ also occurs for other groups, and whether this can also be compensated by a more efficient implementation of the operation $T$, is presently unknown.  The implementation of our protocol to the case of collective noise channels associated with continuous groups remains an interesting open problem. 

\section*{Acknowledgements}
We would like to thank Giulio Chiribella for his helpful comments and for pointing out an alternative proof of Lemma~\ref{lem:1}.  Michael Skotiniotis would like to thank the Institute for Theoretical Physics at the University of Innsbruck for their gracious hospitality while this work was being conducted.  The research was funded by the Austrian Science Fund (FWF): Y535-N16, SFB F40-FoQus, P20748-N16, P24273-N16 and the European Union (NAMEQUAM), NSERC Canada, EU-Canada exchange, and USARO.

\appendix

\section{\label{appendsu2}Representation theory of $\mathrm{SU}(2)$}

In this appendix we show how DFS arise in the presence of the most general type of collective noise that is associated with the group $\mathrm{SU}(2)$ on $N$, two-dimensional quantum systems~\footnote{In particular the most general collective noise is given by $\{p_g, \,g\in\mathrm{SU}(2)\}$ with $0\leq p_g<1$ satisfying $\int\, p_g\mathrm{d}g=1$.}.  
%the presence of DFS We demonstrate the use of a DFS with a familiar example. The most general type of collective noise on $N$, two-dimensional quantum systems (see~\cite{BRS07} for more details) is associated with the group $\mathcal{SU}(2)$.  
Since any $U\in \mathrm{SU}(2)$ can be written as $e^{-i\frac{\theta}{2} \vec{n}\cdot\vec{\sigma}}$, where $ \vec{n}$ denotes a three-dimensional vector and $\vec{\sigma}$ denotes the Pauli vector, $U^{\otimes N}=e^{-i\theta \vec{n}\cdot\vec{J}}$, with $\vec{J}=\frac{1}{2} \sum_i \vec{\sigma}_i$ denoting the total angular momentum operator.  Hence, $U^{\otimes N}$ commutes with $J^2=\vec{J}\cdot \vec{J}$ for any $U\in \mathrm{SU}(2)$. Thus, any $U^{\otimes N}$ is block-diagonal in the eigenbasis of $J^2$.  Denote by $\{\ket{J,M,\beta}\equiv\ket{J,M}\otimes\ket{\beta}\}$ an orthonormal basis for the $2^N$-dimensional Hilbert space, $\cH_2^{\otimes N}$, where $\{\ket{J,M}\}_{M=-J}^J$ is the joint eigenbasis of $J^2$ and the $z$-component of the total angular momentum operator, $J_z$, and  $\beta\in(1,\ldots,\alpha^{(J)})$ is a degeneracy (multiplicity) index, with $\alpha^{(J)}$ the number of orthonormal states of total angular momentum $J$ and $J_z=M$. In this basis  $U^{\otimes N}$ may be conveniently written in block diagonal form as $U^{\otimes N}=\oplus_{J}\,U^{(J)}\otimes I_{\alpha^{(J)}}$, where $J$ is the total angular momentum, $U^{(J)}$ are the irreps of $\mathrm{SU}(2)$, and $I_{\alpha^{(J)}}$ denotes the $\alpha^{(J)}$-dimensional identity operator.

Consequently, we can decompose the total Hilbert space, $\cH_2^{\otimes N}$, into orthogonal subspaces, ${\cal H}^{(J)}=\mbox{span}\, \{\ket{J,M,\beta}\}$, where $-J\leq M\leq J$ and $1\leq\beta\leq\alpha^{(J)}$, as ${\cal H}=\oplus_J  H^{(J)}=\oplus_J\cM^{(J)}\otimes\cN^{(J)}$. Here, $\cM^{(J)}$ is the space upon which the irrep $U^{(J)}$ of $\mathrm{SU}(2)$ acts, and $\cN^{(J)}$ is the multiplicity space upon which the trivial representation, $I_{\alpha^{(J)}}$, of $\mathrm{SU}(2)$ acts.  Moreover, as the dimension of the irrep $U^{(J)}$ coincides with $\mathrm{dim}(\cM^{(J)})$ given by $2J+1$, all irreps, $U^{(J)}$, are inequivalent.  It follows that the only irrep of trivial dimension is $U^{(J=0)}$ which, using the rules for addition of angular momenta, occurs in the decomposition of $U^{\otimes N}$ only if $N$ is even. Hence, for $N$ even, the sector $\cH^{(J=0)}$ is a decoherence-free subspace. For  $J>0$, the irreps, $U^{(J)}$ are of non-trivial dimension, and the sectors $\cH^{(J)}$ are no longer decoherence-free. For $J>0$ a decoherence-free subsystem exists if $\alpha^{(J)}>1$.

The smallest number of spin-1/2 systems that admits a non-trivial noiseless encoding for a general $U\in\mathrm{SU}(2)$ occurs for the case of three qubits.  The representation, $U^{\left(\frac{1}{2}\right)}_{(\theta,\phi,\psi)}$, where $(\theta,\phi,\psi)$ are the Euler angles, acting on a two-dimensional Hilbert space is given by
\begin{equation}
U^{\left(\frac{1}{2}\right)}_{(\theta,\phi,\psi)}=\left(\begin{array}{cc}
e^{-i\frac{1}{2}(\theta+\psi)}\cos(\phi/2) & -e^{-i\frac{1}{2}(\theta-\psi)}\sin(\phi/2)\\
e^{i\frac{1}{2}(\theta-\psi)}\sin(\phi/2) & e^{i\frac{1}{2}(\theta+\psi)}\cos(\phi/2)\end{array}\right).
\label{8}
\end{equation}
As mentioned above, the representation $U^{\left(\frac{1}{2}\right)\otimes 3}_{(\theta,\phi,\psi)}$ can be decomposed into orthogonal sectors, labeled by the total angular momentum $J$, as $U^{\left(\frac{1}{2}\right)\otimes 3}_{(\theta,\phi,\psi)}=U^{\left(\frac{3}{2}\right)}_{(\theta,\phi,\psi)}\bigoplus U^{\left(\frac{1}{2}\right)}_{(\theta,\phi,\psi)}\otimes I_2$, where $I_2$ is the two-dimensional identity operator and the representation $U^{\left(\frac{3}{2}\right)}_{(\theta,\phi,\psi)}$ is given by
\begin{equation}
\bra{\frac{3}{2},m'}U^{\left(\frac{3}{2}\right)}_{(\theta,\phi,\psi)}\ket{\frac{3}{2},m}=e^{-i m'\theta}\,d^{\left(\frac{3}{2}\right)}_{m',m}(\phi)\,e^{-i m\psi}.
\label{9}
\end{equation}
The matrix $d^{\left(\frac{3}{2}\right)}(\phi)$ in Eq.~\eqref{9} is the Wigner small-$d$ matrix given by
\begin{align}\nonumber
d^{\left(\frac{3}{2}\right)}_{3/2,3/2}(\phi)&=\frac{1+\cos\phi}{2}\cos\frac{\phi}{2} \\ \nonumber
d^{\left(\frac{3}{2}\right)}_{3/2,1/2}(\phi)&=-\sqrt{3}\frac{1+\cos\phi}{2}\sin\frac{\phi}{2}\\ \nonumber
d^{\left(\frac{3}{2}\right)}_{3/2,-1/2}(\phi)&=\sqrt{3}\frac{1-\cos\phi}{2}\cos\frac{\phi}{2}\\ \nonumber
d^{\left(\frac{3}{2}\right)}_{3/2,-3/2}(\phi)&=-\frac{1-\cos\phi}{2}\sin\frac{\phi}{2}\\ \nonumber
d^{\left(\frac{3}{2}\right)}_{1/2,1/2}(\phi)&=\frac{3\cos\phi-1}{2}\cos\frac{\phi}{2}\\
d^{\left(\frac{3}{2}\right)}_{1/2,-1/2}(\phi)&=-\frac{3\cos\phi+1}{2}\sin\frac{\phi}{2},
\label{10}
\end{align}
where the matrix elements, $d^{\left(\frac{3}{2}\right)}_{m',m}(\phi)$,  satisfy  the relation $d^{\left(\frac{3}{2}\right)}_{m,m'}(\phi)=(-1)^{m-m'}d^{\left(\frac{3}{2}\right)}_{m',m}(\phi)=d^{\left(\frac{3}{2}\right)}_{-m',m}(\phi)$.  Consequently, the total Hilbert space, $\cH_2^{\otimes 3}$, decomposes into orthogonal sectors, labelled by the total angular momentum quantum number $J$, with orthonormal basis vectors
\begin{eqnarray}\nonumber
\ket{J=\frac{3}{2},M=\frac{3}{2}}&=&\ket{000}\\ \nonumber
\ket{J=\frac{3}{2},M=\frac{1}{2}}&=&\frac{1}{\sqrt{3}}\left(\ket{001}+\ket{010}+\ket{100}\right)\\ \nonumber
\ket{J=\frac{3}{2},M=-\frac{1}{2}}&=&\frac{1}{\sqrt{3}}\left(\ket{110}+\ket{101}+\ket{011}\right)\\ \nonumber
\ket{J=\frac{3}{2},M=-\frac{3}{2}}&=&\ket{111}\\ \nonumber
\ket{J=\frac{1}{2},M=\frac{1}{2},\beta=0}&=&\frac{1}{\sqrt{2}}\left(\ket{100}-\ket{010}\right)\\ \nonumber
\ket{J=\frac{1}{2},M=-\frac{1}{2},\beta=0}&=&\frac{1}{\sqrt{2}}\left(\ket{011}-\ket{101}\right)\\ \nonumber
\ket{J=\frac{1}{2},M=\frac{1}{2},\beta=1}&=&\sqrt{\frac{2}{3}}\ket{001}-\frac{\ket{010}+\ket{100}}{\sqrt{6}}\\ \nonumber
\ket{J=\frac{1}{2},M=-\frac{1}{2},\beta=1}&=&\sqrt{\frac{2}{3}}\ket{110}-\frac{\ket{101}+\ket{011}}{\sqrt{6}},\\
\label{11}
\end{eqnarray}
where the degeneracy index, $\beta$, keeps track of whether $J=1/2$ arose due to the coupling of the first two qubits in a spin-1 or  spin-0 state of total angular momentum. Thus, the sector $\cH^{\left(J=\frac{1}{2}\right)}$ contains a two-dimensional DFS. Defining the logical basis
$
\ket{0_L}\equiv c_1\ket{J=\frac{1}{2},M=\frac{1}{2},\beta=0}+c_2\ket{J=\frac{1}{2},M=\frac{-1}{2},\beta=0},\\
\ket{1_L}\equiv d_1\ket{J=\frac{1}{2},M=\frac{1}{2},\beta=1}+d_2\ket{J=\frac{1}{2},M=\frac{-1}{2},\beta=1},
$
where $|c_1|^2+|c_2|^2=1$ and $|d_1|^2+|d_2|^2=1$, one can transmit one logical qubit noiselessly through the channel.

\section{\label{append1} Proof of Theorem~\ref{thm:1}}

In this appendix we provide a detailed proof of Theorem~\ref{thm:1} in Sec.~\ref{sec:3} regarding the construction of a set of token states, $S^{(U)}_{(r,\ket{\psi})}$. First, we show that for a representation, $U$, of a group, $G$, that is isomorphic to $G$, there exists an integer, $r$, such that $U^{\otimes r}$ contains the regular representation, $\cR$, of $G$ as a sub-representation. Next we show that there exists an $r$ and a state, $\ket{\psi}$, such that the set of states $S^{(U)}_{(r,\ket{\psi})}$ satisfies conditions (1) and (2) in Sec.~\ref{sec:3} if and only if $U^{\otimes r}$ contains the regular representation, $\cR$, of $G$ as a sub-representation.  Finally, we demonstrate that the state $\ket{\psi}\in\cH_d^{\otimes r}$ in the definition of $S^{(U)}_{(r,\ket{\psi})}$ can be chosen to be of the form given by Eq.~\eqref{19}.

\begin{lemma}
Let $U$ be a representation of $G$ that is isomorphic to $G$.  Then there exists a finite integer $r$, such that $U^{\otimes r}$ contains the regular representation as a sub-representation.
\label{lem:1}
\end{lemma}
To prove Lemma~\ref{lem:1} we will make use of the following theorem, whose proof can be found in~\cite{Burnside}.  An alternative proof, as well as several bounds on the integer, $r$, can also be found in~\cite{C06}~\footnote{We became aware of the fact that Lemma~\ref{lem:1} had already been proved by~\cite{C06} after this work was completed.}.
\begin{theorem}
Let $U$ be a representation of $G$ that is isomorphic.  Then, there exists an integer $n$, such that $U^{\otimes n}$ contains every irrep of $G$ at least once.
\label{burn}
\end{theorem}

\begin{proof}{(Lemma~\ref{lem:1}).}
Write $U$ as the sum of inequivalent irreps, $U^{(\lambda)}$,
\begin{equation}
U=\sum_\lambda\alpha^{(\lambda)} U^{(\lambda)},
\label{A1}
\end{equation}
where $\alpha^{(\lambda)}$ is the multiplicity of irrep $U^{(\lambda)}$.  The character of the representation $U$, on the conjugacy class $[g_i]$, $\chi_{[g_i]}$, is given by
\begin{equation}
\chi_{[g_i]}=\sum_\lambda \alpha^{(\lambda)} \chi^{(\lambda)}_{[g_i]}.
\label{A2}
\end{equation}
Since $G$ is a finite group, and $U$ is isomorphic to $G$, it follows from Theorem~\ref{burn} that there exists an integer, $n$, such that $U^{\otimes n}$ contains every irrep of $G$ at least once.  Defining $\Gamma=\bigoplus_{\lambda=1}^s U^{(\lambda)}$, where $s$ denotes as before the number of inequivalent irreps, and using Theorem~\ref{burn} it follows that
\begin{equation}
U^{\otimes n}=\Gamma\bigoplus V,
\label{A3}
\end{equation}
where $V$ is a representation of $G$.  Now consider the decomposition of $U^{\otimes nm}$ where $m$ is some integer.  This may be written as
\begin{equation}
U^{\otimes nm}=\left(\Gamma\bigoplus V\right)^{\otimes m}.
\label{A4}
\end{equation}

If two matrices $A$ and $B$ are block diagonal, then $A\otimes B$ is also block diagonal, and if $A$ and $B$ are representations of $G$,  then so is $A\otimes B$. Moreover, $A\otimes B$ is reducible, so that each block of $A\otimes B$ can be reduced further into sub-blocks. Consider only the block  $\Gamma^{\otimes m}$ from Eq.~\eqref{A4}.  This block can be written as
\begin{align}\nonumber
\Gamma^{\otimes m}&=\left(\bigoplus_{\lambda=1}^sU^{(\lambda)}\right)\otimes \Gamma^{\otimes m-1}\\
&=\bigoplus_{\lambda=1}^sU^{(\lambda)}\otimes\left(\bigoplus_{\lambda'=1}^s U^{(\lambda')}\right)^{\otimes m-1}.
\label{A5}
\end{align}
Each block, labeled by $\lambda$, in Eq.~\eqref{A5} consists of sub-blocks given by $U^{(\lambda)}\bigotimes_{i=1}^{m-1}U^{(\nu_i)}$, where $\nu_i$ can take any value from the set of irrep labels $\{1,\ldots,s\}$. One such sub-block is the one where all $\nu_i=\lambda$, that is $U^{(\lambda)\otimes m}$.  If $U^{(\lambda)}$ is isomorphic to $G$, then by Theorem~\ref{burn} there exists an integer $m$ such that $U^{(\lambda)\otimes m}=\Gamma\bigoplus V'$, where $V'$ is a representation of $G$~\footnote{For irreps that are not isomorphic to $G$ one can produce a similar argument. Let $U^{(\lambda)}$, be one such irrep. Then, $U^{(\lambda)}\otimes \Gamma=\bigoplus_\lambda a^{(\lambda)}U^{(\lambda)}$, and $U^{(\lambda)}\otimes \Gamma^{\otimes m-1}= \left(\bigoplus_\lambda a^{(\lambda)}U^{(\lambda)}\right)\otimes \Gamma^{\otimes m-2}$.  Using a similar argument as above it follows that for suitable integer $m$ the block $U^{(\lambda)}\otimes \Gamma^{\otimes m}$ will contain $\Gamma$ as a sub-representation.}. Hence
\begin{equation}
\Gamma^{\otimes m}= \kappa \Gamma\bigoplus T,
\label{A6}
\end{equation}
where $\kappa>1$ is an integer, and $T$ is a representation of $G$. This process can be carried as far as we like increasing the multiplicity of every irrep as much as we like. Note that this is not the most efficient way  to increase the multiplicities of every irrep,  as we focused only on the sub-block, $U^{(\lambda)\otimes m}$, in Eq.~\eqref{A5} and neglected all other irrep products.

As the dimensions of all irreps are finite, there exists an $r_\lambda$ for each $\lambda$ such that the irrep $U^{(\lambda)}$ occurs at least $\mathrm{dim}(U^{(\lambda)})$ times in $U^{\otimes r_\lambda}$.  Choosing $r=\mathrm{max}\{r_\lambda\}$ ensures that $U^{\otimes r}$ contains the regular representation, $\cR$, as a sub-representation.  In Appendix~\ref{append2} we illustrate how one can compute $r$ for some examples.
\end{proof}

We are now ready to prove that there exists an $r$ and a state $\ket{\psi}\in\cH^{\otimes r}$ such that the set $S^{(U)}_{(r,\ket{\psi})}$ satisfies conditions (1) and (2) in Sec.~\ref{sec:3}, if and only if the representation $U^{\otimes r}$ contains the regular representation, $\cR$ of $G$, as a sub-representation (Theorem 1).
\begin{proof}(Theorem~\ref{thm:1}).
To prove the backward implication assume that $U^{\otimes r}=\cR\bigoplus_\lambda \alpha_\lambda U^{(\lambda)}$. The total Hilbert space, $\cH^{\otimes r}$, decomposes as $\cH^{\otimes r}=\cH_{\cR}\bigoplus_\lambda \cH^{(\lambda)}$, where $\cH_{\cR}$ is the Hilbert space on which the regular representation is acting.  Denote by $\{\ket{\psi(g_i)};\, i=0,\ldots,\lvert G\rvert-1\}$ the first $\lvert G\rvert$ standard basis vectors in $\cH^{\otimes r}$, i.e.~the computational basis of $\cH_{\cR}$ embedded in the $2^r$-dimensional Hilbert space $\cH^{\otimes r}$.  From the definition of the regular representation it follows that
\begin{equation}
U_{g_k}^{\otimes r}\ket{\psi(g_i)}=\ket{\psi(g_k\cdot g_i)}=\ket{\psi(g_l)},
\label{A7}
\end{equation}
where $g_k\cdot g_i$ is the group product, and $l\in (0,\ldots,\lvert G\rvert-1)$ is such that $g_k\cdot g_i=g_l$.  Thus, the set of states $\{\ket{\psi(g_i)};\, i=0,\ldots, \lvert G\rvert-1\}$ satisfies properties (1) and (2).

To prove the forward implication assume that the set of states, $S^{(U)}_{(r,\ket{\psi})}$, where $\ket{\psi}\in\cH^{\otimes r}$, satisfies properties (1) and (2).  Define
\begin{align}\nonumber
P_G&\equiv\frac{1}{|G|}\sum_{g_i\in G}U_{g_i}^{\otimes r}\ketbra{\psi}U_{g_i}^{\otimes r\,\dagger}\\
&=\frac{1}{\lvert G\rvert}\sum_{g_i\in G}\ketbra{\psi(g_i)}.
\label{A8}
\end{align}
Then by property (2)
\begin{align}\nonumber
U_{g_k}^{\otimes r}P_GU_{g_k}^{\otimes r\,\dagger}&=\frac{1}{|G|}\sum_{g_i\in G}U_{g_k}^{\otimes r}\ketbra{\psi(g_i)}U_{g_k}^{\otimes r\,\dagger}\\
&=\frac{1}{|G|}\sum_{g_i\in G}\ketbra{\psi(g_k\cdot g_i)}.
\label{A9}
\end{align}
Denoting $g_k\cdot g_i=g_l\in G$, Eq.~\eqref{A9} can be written as
\begin{equation}
U_{g_k}^{\otimes r}P_GU_{g_k}^{\otimes r\,\dagger}=\frac{1}{|G|}\sum_{g_l\in G}\ketbra{\psi(g_l)}=P_G.
\label{A10}
\end{equation}
Since $p_g$ in Eq.~\eqref{A8} is the Haar measure, it follows form Schur's first lemma that $P_G$ is a multiple of the $G$-dimensional identity.   Using Eq.~\eqref{6}, Eq.~\eqref{A8} may also be written as
\begin{equation}
P_G=\sum_\lambda \left(\cD_{\cM^{(\lambda)}}\otimes \cI_{\cN^{(\lambda)}}\right)\circ \cP^{(\lambda)}[\ketbra{\psi}],
\label{A11}
\end{equation}
where $\lambda$ labels the irreps present in the decomposition of $U^{\otimes r}$, $\cD$ is the completely depolarizing map, $\cD(A)=\frac{\mbox{tr}(A)}{\mathrm{dim}(\cH)}\one, \,\forall A\in\cB(\cH)$, $\cI$ is the identity map, and $\cP^{(\lambda)}(A)=\Pi_\lambda A\Pi_\lambda$, where $\Pi_\lambda$ is the projector onto the space $\cH^{(\lambda)}$.

Write $\ket{\psi}=\sum_\lambda c_\lambda\ket{\psi^{(\lambda)}}$, where $c_\lambda\in\C$ satisfy $\sum_\lambda\lvert c_\lambda\rvert^2=1$, and $\ket{\psi^{(\lambda)}}=\Pi^{(\lambda)}\ket{\psi}$. Using the Schmidt decomposition, and defining $\left\{\ket{\xi^{(\lambda)}_n}\right\}$ and $\left\{\ket{\zeta^{(\lambda)}_n}\right\}$ as orthonormal basis for $\cM^{(\lambda)}$ and $\cN^{(\lambda)}$ respectively, we may write
\begin{equation}
\ket{\psi^{\lambda}}=\sum_{n=1}^{\tilde{d}_\lambda}\mu^{(\lambda)}_n\ket{\xi^{(\lambda)}_n}\ket{\zeta^{(\lambda)}_n},
\label{A12}
\end{equation}
where $\tilde{d}_\lambda\leq \mathrm{min}\{d_\lambda=\mathrm{dim}(\cM^{(\lambda)}),\,\mathrm{dim}(\cN^{(\lambda)})\}$ and $0\neq \mu^{(\lambda)}_n\in\mathbb{R} $ are the \defn{Schmidt coefficients}.  Substituting Eq.~\eqref{A12} into Eq.~\eqref{A11} gives
\begin{align}\nonumber
P_G&=\sum_\lambda |c_\lambda|^2\sum_{n,n'=1}^{\tilde{d}_\lambda}\mu^{(\lambda)}_n\mu^{(\lambda)*}_{n'}\left(\cD_{\cM^{(\lambda)}}\left[\ket{\xi^{(\lambda)}_n}\bra{\xi^{(\lambda)}_{n'}}\right]\right.\\ \nonumber
&\left.\otimes\, \cI\left[\ket{\zeta^{(\lambda)}_n}\bra{\zeta^{(\lambda)}_{n'}}\right]\right)\\
&=\sum_\lambda |c_\lambda|^2\sum_{n=1}^{\tilde{d}_\lambda}|\mu^{(\lambda)}_n|^2\frac{\one_{d_\lambda}}{d_\lambda}\otimes \ketbra{\zeta^{(\lambda)}_n}.
\label{A13}
\end{align}
Computing the rank on both sides of Eq.~\eqref{A13}, and using the notation $\rho_{\cN^{(\lambda)}}=\mathrm{tr}_{\cM^{(\lambda)}}[\ketbra{\psi^{(\lambda)}}]$, one obtains
\begin{equation}
|G|=\sum_\lambda d_\lambda\, \mathrm{rk}\left(\rho_{\cN^{(\lambda)}}\right).
\label{A14}
\end{equation}
As $|G|=\sum_{\lambda} d_\lambda^2$ (see Eq(\ref{18})), and $\mathrm{rk}\left(\rho_{\cN^{(\lambda)}}\right)=\tilde{d}_\lambda\leq d_\lambda$, $\forall \lambda$, it follows that $\mathrm{rk}\left(\rho_{\cN^{(\lambda)}}\right)=\tilde{d}_\lambda=d_\lambda$.  Hence $\tilde{d}_\lambda= \mathrm{min}\{d_\lambda=\mathrm{dim}(\cM^{(\lambda)}),\,\mathrm{dim}(\cN^{(\lambda)})\}$ and therefore the multiplicity of each irrep, ${\mathrm{dim}}(\cN^{(\lambda)})$, occurs a number of times greater than or equal to its dimension, $d_\lambda$.  Thus, $U^{\otimes r}$ contains the regular representation, $\cR$, of $G$ as a sub-representation.  This completes the proof.
\end{proof}

In order to construct the set of token states, $S^{(U)}_{(r,\ket{\psi})}$, for an isomorphic representation, $U$, and an $r$ chosen such that $U^{\otimes r}$ contains the regular representation, one can choose the state $\ket{\psi}\in\cH_d^{\otimes r}$ to be of the form in Eq.~\eqref{19} as we now show.  
Since $\mathrm{tr}(P_G)=1$ we have $P_G=\frac{1}{|G|} \one$, where $\one$ is the $\lvert G\rvert$-dimensional identity operator.  Writing $\ket{\psi}=\sum_\lambda c_\lambda \ket{\psi^{(\lambda)}}$, with $\ket{\psi^{(\lambda)}}$ given by the Schmidt decomposition, Eq.~\eqref{A12}, and using Eq.~\eqref{A11} we have
\begin{align}\nonumber
P_G=\frac{1}{|G|}\sum_\lambda I_{\cM^{(\lambda)}}\otimes I_{\cN^{(\lambda)}}&=\sum_\lambda |c_\lambda|^2\left(\cD_{\cM^{(\lambda)}}\otimes\cI_{\cN^{(\lambda)}}\right)\\
&\left[\ketbra{\psi^{(\lambda)}}\right].
\label{A15}
\end{align}
As both $\cD$ and $\cI$ are trace-preserving quantum operations, looking at a single sector, $\lambda$, and computing the trace on both $\cM^{(\lambda)}$ and $\cN^{(\lambda)}$ one obtains
\begin{equation}
\frac{1}{|G|}d_\lambda^2=|c_\lambda|^2.
\label{A16}
\end{equation}

Inserting this value for $|c_\lambda|^2$ in the expression given in Eq. (\ref{A13}) for $P_G$ leads to
\begin{equation}
P_G=\frac{1}{|G|}\sum_\lambda I_{\cM^{(\lambda)}}\otimes \sum_{n=1}^{d_\lambda} |\mu_n^{(\lambda)}|^2 \ketbra{\zeta^{(\lambda)_n}}.
\label{A17}
\end{equation}
In addition,
\begin{align}\nonumber
P_G^2&=\sum_{\lambda,\lambda'}\lvert c_\lambda\rvert^2\lvert c_{\lambda'}\rvert^2\sum_{n,m}\lvert\mu^{(\lambda)}_n\rvert^2\lvert \mu^{(\lambda')}_m\rvert^2\bracket{\xi^{(\lambda)}_n}{\xi^{(\lambda')}_m}\\ \nonumber
&\times \bracket{\zeta^{(\lambda)}_n}{\zeta^{\lambda')}_m}\frac{\ket{\xi^{(\lambda)}_n}\bra{\xi^{(\lambda')}_m}}{d_\lambda d_\lambda'}\otimes \ket{\zeta^{(\lambda)}_n}\bra{\zeta^{\lambda')}_m}\\
&=\sum_{\lambda}\lvert c_\lambda\rvert^4\sum_n\lvert \mu^{(\lambda)}_n\rvert^4\frac{\ketbra{\xi^{(\lambda)}_n}}{d_\lambda^2}\otimes\ketbra{\zeta^{\lambda)}_n},
\label{A18}
\end{align}
where we have used the fact that $\bracket{\xi^{(\lambda)}_n}{\xi^{(\lambda')}_m}=\delta_{\lambda,\lambda'}\delta_{n,m}$.  As $P_G^2=\frac{1}{|G|}P_G$, equating the terms of Eq.~\eqref{A18} and Eq.~\eqref{A13} yields
\begin{align}\nonumber
&\frac{\lvert c_\lambda\rvert^4\lvert\mu^{(\lambda)}_n\rvert^4}{d_\lambda^2}=\frac{\lvert c_\lambda\rvert^2\lvert\mu^{(\lambda)}_n\rvert^2}{|G|d_\lambda}\\
&\frac{\lvert c_\lambda\rvert^2\lvert\mu^{(\lambda)}_n\rvert^2}{d_\lambda}\left(\frac{\lvert c_\lambda\rvert^2\lvert\mu^{(\lambda)}_n\rvert^2}{d_\lambda}-\frac{1}{|G|}\right)=0.
\label{A19}
\end{align}
Hence, using Eq.~\eqref{A16},
$\mu^{(\lambda)}_n=d_\lambda^{-1/2}$ $\forall \lambda,n$ and thus, the state $\ket{\psi}\in\cH_d^{\otimes r}$ can be chosen as
\begin{equation}
\ket{\psi}=\sum_\lambda\sqrt{\frac{d_\lambda}{|G|}}\sum_n^{d_\lambda}\ket{\xi^{(\lambda)}_n}\ket{\zeta^{(\lambda)}_n}.
\label{A20}
\end{equation}

\section{\label{append2} Calculating the number of auxiliary systems $r$}

In this appendix we show how to calculate the number of auxiliary systems, $r$, required in our protocol  such that $U^{\otimes r}$, where $U$ is a isomorphic representation of a group $G$, contains the regular representation, $\cR$, of $G$.   In particular, we calculate $r$ in the case where the collective noise of the channel is associated with the cyclic group $\mathbb{Z}_3$, and the symmetric group on three  elements, $S_3$.

As a first step we review how the multiplicity of the irreps can be computed (see Sec. II). Let $U$ be a isomorphic representation of a finite group, $G$.  Then by Lemma~\ref{lem:1}, there exists an integer $n$, such that $U^{\otimes n}$ contains every irrep of $G$ at least once.  That is
\begin{equation}
U^{\otimes n}=\sum_\lambda \gamma^{(\lambda)}_n U^{(\lambda)},
\label{B1}
\end{equation}
where $\gamma^{(\lambda)}_n\geq1$ is the number of times the irrep $U^{(\lambda)}$ appears. Since $\mathrm{tr}(U_{g_i}^{\otimes n})=(\mathrm{tr}(U_{g_i}))^n$, it follows that
\begin{equation}
(\chi_{[g_i]})^n=\sum_\lambda \gamma^{(\lambda)}_n \chi^{\lambda}_{[g_i]}.
\label{B2}
\end{equation}
As explained in Sec. II the multiplicities, $\gamma^{(\lambda)}_n$, can be easily computed using Eq. (\ref{13}). We obtain
\begin{equation}
\sum_{i=1}^s|[g_i]|\chi^{(\lambda^\prime)*}_{[g_i]}(\chi_{[g_i]})^n=|G|\gamma^{(\lambda^\prime)}_n.
\label{B3}
\end{equation}

Thus, given the compound character of the representation $U$, as well as the compound characters of all the irreducible representations of $G$ and $|[g_i]|$, one can compute the multiplicities $\gamma^{(\lambda)}_n$ for any integer $n$. To ensure that the regular representation is contained in $U^{\otimes n}$, $\gamma^{(\lambda)}_n\geq d_\lambda$ must hold. We are going to show next how large $n$ must be chosen to ensure the validity of the condition above for some examples.

\subsection{Example 1:  $\mathbb{Z}_3$}
For our first example we consider a channel whose collective noise is given by the two dimensional representation of $\mathbb{Z}_3$ of Eq.~\eqref{27}.  Note that this representation is isomorphic to $\mathbb{Z}_3$. Even though we have already shown in Sec. III that if $U$ is a $d$--dimensional representation of $\mathbb{Z}_N$ it is sufficient to chose $r=\lceil \frac{N-1}{d-1}\rceil$, we explicitly compute $r$ here for the case $d=2,N=3$. The character table for $\mathbb{Z}_3$ is given in Table~\ref{tbl2}, and the regular representation, $\cR$, of $\mathbb{Z}_3$ is given in Eq.~\eqref{28}.

The compound character of the representation given in Eq.~\eqref{27} can be easily computed to be
\begin{equation}
\chi=\left(\begin{array}{c} 2\\-\omega^2\\-\omega\end{array}\right),
\label{B4}
\end{equation}
where $\omega=e^{\frac{i2\pi}{3}}$.  Using Eq.~\eqref{B4}, Table~\ref{tbl2}, and the fact that $|[g_i]|=1,\, \forall i\in(0,1,2)$,  Eq.~\eqref{B3} reads, for $n=1$,
\begin{align}\nonumber
\gamma^{(0)}_1&=\frac{1}{3}\left(2-\omega^2-\omega\right)=1\\  \nonumber
\gamma^{(1)}_1&=\frac{1}{3}\left(2+\omega^2(-\omega^2)+\omega(-\omega)\right)=1\\
\gamma^{(2)}_1&=\frac{1}{3}\left(2+\omega(-\omega^2)+\omega^2(-\omega)\right)=0.
\label{B5}
\end{align}
Thus, the decomposition of $U$, given by Eq.~\eqref{27}, into irreps is $U=U^{(0)}\oplus U^{(1)}$. Using Eq.~\eqref{B3} with  $n=2$ we obtain
\begin{align}\nonumber
\gamma^{(0)}_2&=\frac{1}{3}\left(4+(-\omega^2)^2+(-\omega)^2\right)=1\\ \nonumber
\gamma^{(1)}_2&=\frac{1}{3}\left(4+\omega^2(-\omega^2)^2+\omega(-\omega)^2\right)=2\\
\gamma^{(2)}_2&=\frac{1}{3}\left(4+\omega(-\omega^2)^2+\omega^2(-\omega)^2\right)=1.
\label{B6}
\end{align}
Thus $U^{\otimes 2}=U^{(0)}\oplus 2U^{(1)}\oplus U^{(2)}$.  As all the irreps are one-dimensional, and they all appear at least once, $U^{\otimes 2}$ contains the regular representation.

\subsection*{Example 2:  $S_3$}
As a second example we consider a channel whose collective noise is described by the group $S_3$, the symmetric group on three elements.  Using cycle notation, the elements of $S_3$ are
\begin{equation}
\left\lbrace(1)(2)(3), (123), (132), (12)(3), (13)(2), (23)(1)\right\rbrace,
\label{B13}
\end{equation}
where $(ab)(c)$ denotes the permutation where symbol $c$ remains fixed and symbols $a, b$ are interchanged ($a\rightarrow b\rightarrow a,\, c$), and $(abc)$ denotes the cyclic permutation where $a$ moves to the place of $b$, $b$ moves to the place of $c$, and $c$ moves to the place of $a$ ($a\rightarrow b\rightarrow c\rightarrow a$).  It is known that there is a one-to-one correspondence between the conjugacy classes of the symmetric group and the number of different cycle structures of the group~\cite{Sternberg}.  Hence, $S_3$ has three conjugacy classes; $[(1)(2)(3)]$ which contains only one element, $[(123)]$ which contains two elements, and $[(12)(3)]$ containing three elements.  Furthermore, we note that $S_3$ can be generated by two elements such as $(123)$ and $(12)(3)$.

As $S_3$ has three conjugacy classes, there exist three inquivalent irreps of $S_3$, two one-dimensional irreps, given by $U^{(0)}_{g_i}=1$ for all $g_i\in S_3$ and
\begin{equation}
U^{(1)}_{g_i}=\left\{\begin{array}{l l}
                      1 & \quad \mbox{if $g_i=(123)$}\\
                      -1 & \quad \mbox{if $g_i=(12)(3)$}\\ \end{array} \right.,
\label{B14}
\end{equation}
and a two-dimensional irrep, $U^{(2)}$, given by
\begin{equation}
U^{(2)}_{(123)}=\left(\begin{matrix}
-\frac{1}{2}&-\frac{\sqrt{3}}{2}\\\frac{\sqrt{3}}{2}&-\frac{1}{2}\end{matrix}\right),\,U^{(2)}_{(12)(3)}=\left(\begin{matrix}1&0\\0&-1\end{matrix}\right),
\label{B15}
\end{equation}
where the remaining elements of $S_3$ are generated by $U_{(12)(3)}^aU_{(123)}^b$, with $a\in(0,1)$, and $b\in(0,1,2)$. The character table for $S_3$ is given in Table~\ref{tbl5}.
\begin{table}[htb]
\caption{The character table for $\mathbb{S}_3$}
\label{tbl5}
\begin{tabular}{c|c|c|c}
$\chi$&$[(1)(2)(3)]$&$[(123)]$&$[(12)(3)]$\\
\hline
$U^{(0)}$&1&1&1\\
\hline
$U^{(1)}$&1&1&-1\\
\hline
$U^{(2)}$&2&-1&0
\end{tabular}
\end{table}

The regular representation, $\cR$, of $S_3$ is
\begin{align}
\cR_{(123)}&=\left(\begin{matrix}1&0&0&0&0&0\\
					        0&1&0&0&0&0\\
						0&0&-\frac{1}{2}&-\frac{\sqrt{3}}{2}&0&0\\
						0&0&\frac{\sqrt{3}}{2}&-\frac{1}{2}&0&0\\
						0&0&0&0&-\frac{1}{2}&-\frac{\sqrt{3}}{2}\\
						 0&0&0&0&\frac{\sqrt{3}}{2}&-\frac{1}{2}\end{matrix}\right),\\
\cR_{(12)(3)}&=\left(\begin{matrix}1&0&0&0&0&0\\
					        0&-1&0&0&0&0\\
						0&0&1&0&0&0\\
						0&0&0&-1&0&0\\
						0&0&0&0&1&0\\
						0&0&0&0&0&-1\end{matrix}\right),
\label{B16}
\end{align}
where the remaining elements are generated by $\cR_{(12)(3)}^a\cR_{(123)}^b$ with $a\in(0,1)$, and $b\in(0,1,2)$.

Suppose that the action of our channel is given by Eq.~\eqref{B15}, i.e.~$U=U^{(2)}$, with compound character given by
\begin{equation}
\chi=\left(\begin{array}{c}2\\-1\\0\end{array}\right).
\label{B17}
\end{equation}
Obviously we have $\gamma^{(0)}_1=\gamma^{(1)}_1=0$ and $\gamma^{(2)}_1=1$.  Using Eq.~\eqref{B3} for $n=2$ we obtain
\begin{align}\nonumber
\gamma^{(0)}_2&=\frac{1}{6}\left(1\cdot1\cdot2^2+2\cdot1\cdot(-1)^2+0\right)=1\\ \nonumber
\gamma^{(1)}_2&=\frac{1}{6}\left(1\cdot1\cdot2^2+2\cdot1\cdot(-1)^2+0\right)=1\\
\gamma^{(2)}_2&=\frac{1}{3}\left(1\cdot2\cdot2^2+2\cdot(-1)\cdot(-1)^2+0\right)=1.
\label{B19}
\end{align}
Hence, $U^{\otimes 2}=U^{(0)}\oplus U^{(1)}\oplus U^{(2)}$.  This is not the regular representation as $U^{(2)}$ is two-dimensional but appears only once.  Evaluating Eq.~\eqref{B3} a third time for $n=3$ yields
\begin{align}\nonumber
\gamma^{(0)}_3&=\frac{1}{6}\left(1\cdot1\cdot2^3+2\cdot1\cdot(-1)^3+0\right)=1\\ \nonumber
\gamma^{(1)}_3&=\frac{1}{6}\left(1\cdot1\cdot2^3+2\cdot1\cdot(-1)^3+0\right)=1\\
\gamma^{(2)}_3&=\frac{1}{3}\left(1\cdot2\cdot2^3+2\cdot(-1)\cdot(-1)^3+0\right)=3.
\label{B20}
\end{align}
Hence $U^{\otimes 3}=U^{(0)}\oplus U{(1)}\oplus 3U^{(2)}$.  As each irrep appears a number of times equal to, or greater than its dimension, $U^{\otimes 3}$ contains the regular representation, $\cR$, of $S_3$.
\bibliography{SKDbib}
\end{document}